\documentclass[sigconf, nonacm, pdfa]{acmart}
\newcommand\vldbdoi{XX.XX/XXX.XX}
\newcommand\vldbpages{XXX-XXX}
\newcommand\vldbvolume{19}
\newcommand\vldbissue{13}
\newcommand\vldbyear{2026}
\newcommand\vldbauthors{\authors}
\newcommand\vldbtitle{\shorttitle} 
\newcommand\vldbavailabilityurl{https://github.com/ForwardStar/VectorMaton}
\newcommand\vldbpagestyle{empty}

\usepackage[a-2b]{pdfx}
\usepackage[ruled,linesnumbered,vlined]{algorithm2e}
\usepackage{tikz}
\usepackage{hyperref}
\usepackage{listings}
\usepackage{enumitem}
\usepackage{subfigure}
\usepackage{color}
\usepackage{pgfplots}
\usepackage{multirow}
\usetikzlibrary{patterns}
\usepackage{balance}
\usepackage{subfiles}
\usepackage{adjustbox}
\usepackage{mdframed}
\usepackage{tabularx}
\usepackage[inkscapeformat=png]{svg}

\newtheorem{definition}{Definition}
\newtheorem{example}{Example} 
 
\newtheorem{theorem}{Theorem} 
\newtheorem{lemma}{Lemma}
\AtBeginDocument{%
  \providecommand\BibTeX{{%
    \normalfont B\kern-0.5em{\scshape i\kern-0.25em b}\kern-0.8em\TeX}}}

\newcounter{siqiang}
\numberwithin{siqiang}{section}

\author{Haoxuan Xie}
\email{haoxuan001@e.ntu.edu.sg}
\affiliation{%
  \institution{Nanyang Technological University}
  \country{Singapore}
}

\author{Siqiang Luo}
\authornote{Corresponding author.}
\email{siqiang.luo@ntu.edu.sg}
\affiliation{%
  \institution{Nanyang Technological University}
  \country{Singapore}
}

\begin{document}
\title{VectorMaton: Efficient Vector Search with Pattern Constraints via an Enhanced Suffix Automaton (Extended Version)}

\renewcommand{\shortauthors}{Xie et al.}
\begin{abstract}
Approximate nearest neighbor search (ANNS) has become a cornerstone in modern vector database systems. Given a query vector, ANNS retrieves the closest vectors from a set of base vectors. In real-world applications, vectors are often accompanied by additional information, such as sequences or attributes, motivating the need for fine-grained vector search with constraints on this auxiliary data. Existing methods support attribute-based filtering or range-based filtering on categorical and numerical attributes, but none supports predicates over sequence attributes. In relational databases, predicates such as \texttt{LIKE} and \texttt{CONTAINS} are fundamental operators for filtering pattern-matching records. As vector databases increasingly adopt SQL-style query interfaces, enabling such pattern predicates over sequence attributes (e.g., texts and biological sequences) alongside vector similarity search becomes essential. In this paper, we formulate a novel problem: given a set of vectors each associated with a sequence, retrieve the nearest vectors whose sequences contain a given query pattern. Specifically, we mainly focus on the substring containment case. To address this challenge, we propose \textsc{VectorMaton}, an automaton-based index that integrates pattern filtering with efficient vector search. Extensive experiments on real-world datasets demonstrate that \textsc{VectorMaton} consistently outperforms all baselines, achieving up to two orders of magnitude higher query throughput with similar or higher accuracy.
\end{abstract}

\settopmatter{printfolios=true}
\maketitle
\pagestyle{\vldbpagestyle}
\thispagestyle{empty}
\begingroup\small\noindent\raggedright\textbf{PVLDB Reference Format:}\\
\vldbauthors. \vldbtitle. PVLDB, \vldbvolume(\vldbissue): \vldbpages, \vldbyear.\\
\href{https://doi.org/\vldbdoi}{doi:\vldbdoi}
\endgroup
\begingroup
\renewcommand\thefootnote{}\footnote{\noindent
This work is licensed under the Creative Commons BY-NC-ND 4.0 International License. Visit \url{https://creativecommons.org/licenses/by-nc-nd/4.0/} to view a copy of this license. For any use beyond those covered by this license, obtain permission by emailing \href{mailto:info@vldb.org}{info@vldb.org}. Copyright is held by the owner/author(s). Publication rights licensed to the VLDB Endowment. \\
\raggedright Proceedings of the VLDB Endowment, Vol. \vldbvolume, No. \vldbissue\ %
ISSN 2150-8097. \\
\href{https://doi.org/\vldbdoi}{doi:\vldbdoi} \\
}\addtocounter{footnote}{-1}\endgroup

\ifdefempty{\vldbavailabilityurl}{}{
\vspace{.3cm}
\begingroup\small\noindent\raggedright\textbf{PVLDB Artifact Availability:}\\
The source code, data, and/or other artifacts have been made available at \url{\vldbavailabilityurl}.
\endgroup
}
\begin{figure}[t]
    \centering
    \includegraphics[width=1.\linewidth]{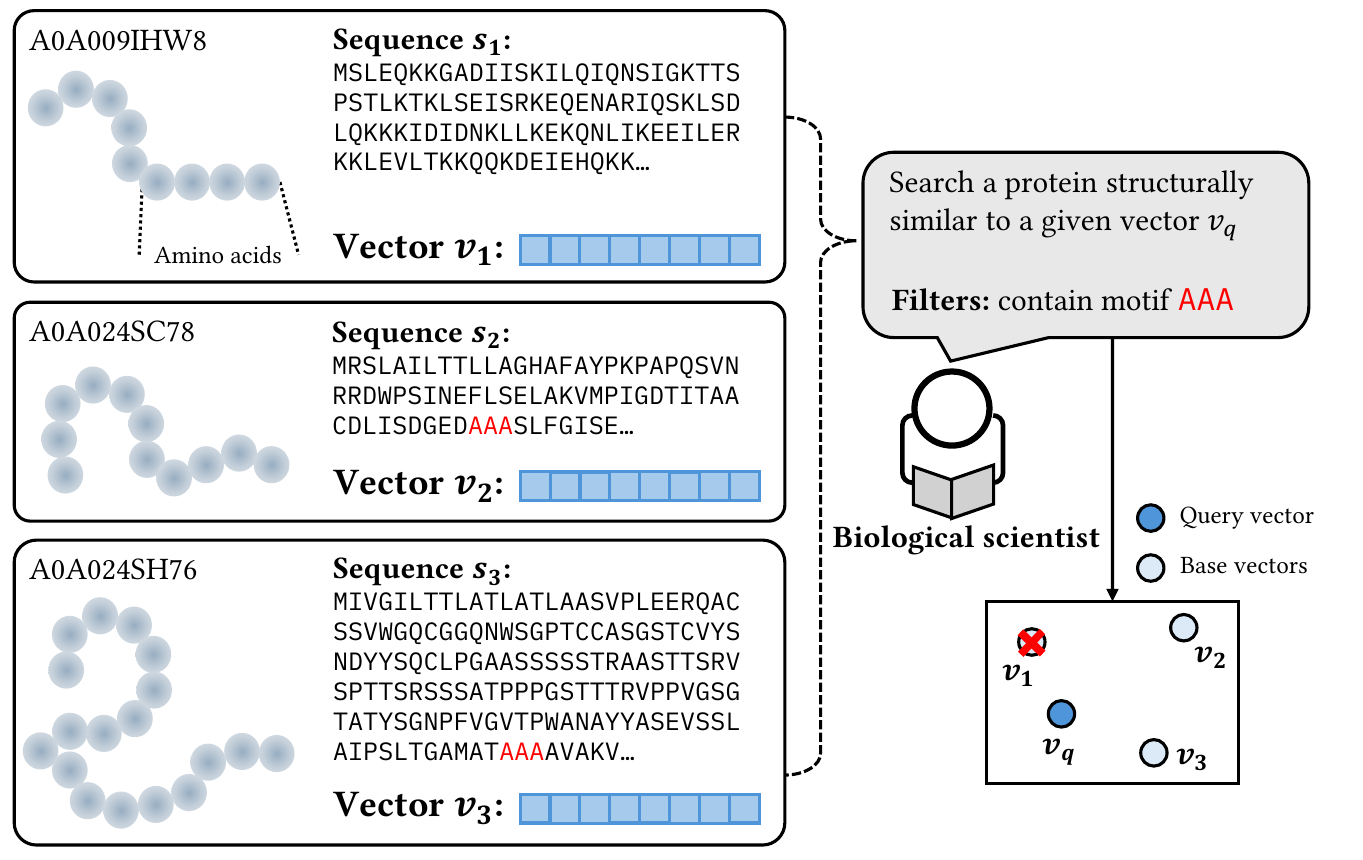}
    \vspace{-8mm}
    \caption{An example of pattern-constrained ANNS in a biological database.}
    \label{fig:example}
\end{figure}
\vspace{-2mm}
\section{Introduction}
\label{sec:intro}
The rapid advancement of deep learning has enabled the transformation of unstructured data (e.g., texts, images and graphs) into high-dimensional vector representations \cite{word2vec,node2vec,10.1109/TPAMI.2013.50,doc2vec}, giving rise to vector databases as a new data management paradigm.
In vector databases, vector search focuses on retrieving the most similar vectors according to a distance or similarity metric.
However, identifying the exact nearest neighbors in high-dimensional vector spaces is computationally expensive due to the curse of dimensionality \cite{ringcovertree}.
As a result, approximate nearest neighbor search (ANNS) has emerged as the dominant approach and attracted significant attention from both academia and industry \cite{lshann,yufeiann,hnsw,10.1145/3709693,nsg,pq,faiss,diskann}.

\begin{figure*}[t]
\centering
\subfigure[\normalsize Comparison of different methods ($m$ is the total sequence length).]{
\begin{minipage}[b]{0.55\linewidth}
\centering
\resizebox{\linewidth}{!}{
\begin{tabular}{l|cc|cc|c}
\hline
\multirow{2}{*}{\textbf{Methods}} 
& \multicolumn{2}{c|}{\textbf{Short pattern}} 
& \multicolumn{2}{c|}{\textbf{Long pattern}}
& \multirow{2}{*}{\textbf{Index Size}} \\
\cline{2-5}
& \textbf{Efficiency} 
& \textbf{Quality} 
& \textbf{Efficiency} 
& \textbf{Quality} 
& \\
\hline
PreFiltering  & Low      & Optimal & Moderate  & Optimal & Small \\
PostFiltering & Moderate & Moderate & Moderate & Low     & Small \\
OptQuery      & High     & High    & High     & High   & $O(m^2)$ \\
\textsc{VectorMaton} & High & High & High & High & $O(m^{1.5})$ \\
\hline
\end{tabular}
}
\label{tab:comparison}
\end{minipage}}
\hfill
\subfigure[\normalsize Impact of pattern length.]{
\centering
\includegraphics[width=.4\linewidth]{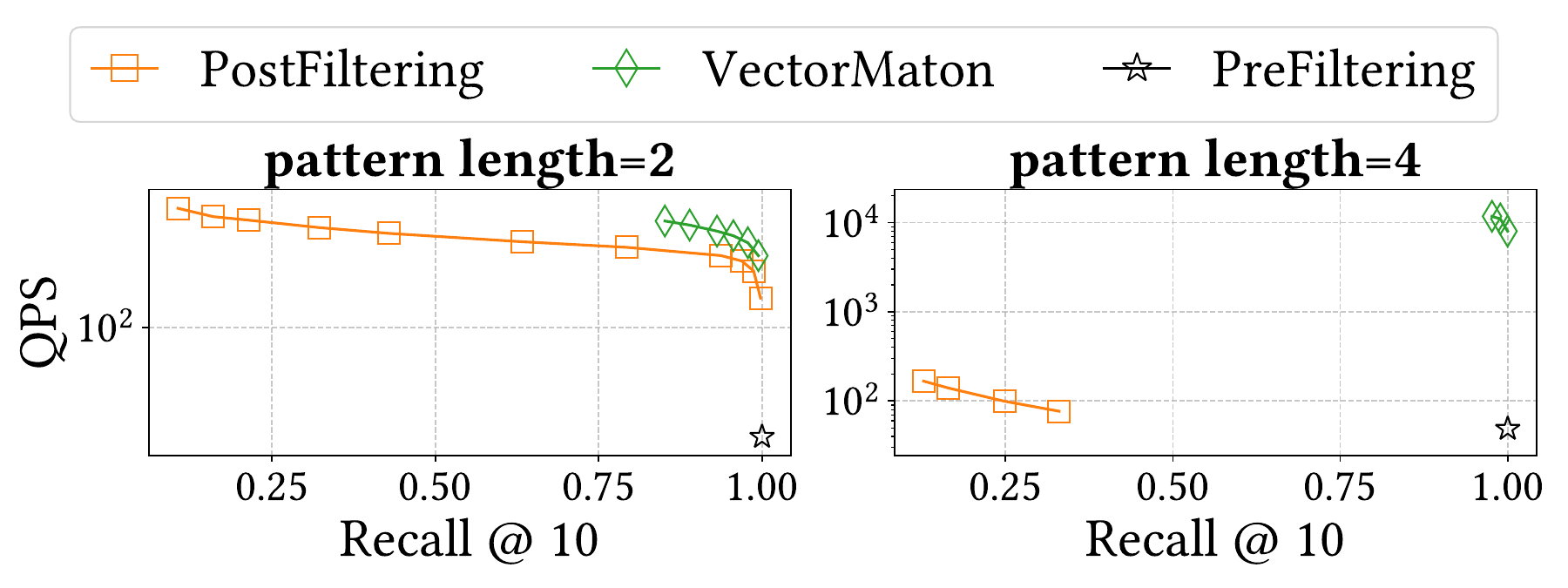}
\label{fig:pattern}
}
\vspace{-5mm}
\caption{Summary of challenges and methods.}
\end{figure*}

\vspace{0.5mm}
\noindent\textbf{Motivation.} In many real-world applications, vector embeddings are accompanied by associated sequences or other structured attributes. For example, in scientific document retrieval, the content of a paper can be embedded into a vector representation, while the paper title is naturally modeled as a sequence associated with the vector~\cite{arxiv-small}. While recent studies have explored constrained approximate nearest neighbor search through \emph{attribute-in} filtering~\cite{unify,attributeannsurvey,filtereddiskann,ivf-subset} and range-based filtering~\cite{rangepq,dynamicrfann,serf,digra,irangegraph}, these techniques primarily focus on numeric or categorical constraints.
In contrast, \textbf{pattern constraints over sequences}, such as substring or motif matching, remain largely unexplored in the context of vector search.
Given a query vector $v$ and a pattern $p$, \emph{ANNS with pattern constraints} aims to retrieve the nearest vectors whose associated sequences contain the pattern $p$. Specifically, we mainly focus on substring containment queries (i.e., pattern is also a sequence). The problem naturally arises in a wide range of real-world applications, including but not limited to the following examples:
\begin{itemize}[leftmargin=0pt,labelsep=0.5mm]
    \item \textbf{Supporting \texttt{LIKE} and \texttt{CONTAINS} predicates in vector SQL.}
    In relational databases, predicates such as \texttt{LIKE} and \texttt{CONTAINS} are fundamental building blocks for SQL queries. These predicates largely involve substring queries, i.e., retrieving data that contains a specific pattern string.
    As vector databases increasingly adopt SQL-like query interfaces, supporting such pattern predicates alongside vector similarity search becomes essential. For example, in a document database, users may not only want to retrieve documents with similar embeddings, but also require the documents to contain a specific keyword. While systems such as Elasticsearch \cite{elasticsearch-vector} and PostgreSQL \cite{pgvector} have added support for hybrid vector similarity and \texttt{LIKE} predicate, they generally rely on either a filter-then-search or a search-then-filter execution strategy. These approaches treat vector search and pattern filtering as separate components and do not provide a unified index structure, leading to low query efficiency and quality. Efficiently integrating pattern constraints with approximate nearest neighbor search remains an open challenge. We would also like to remark that while we mainly focus on substring queries and do not consider the full SQL semantics, our approach can also be extended to support some more complex queries involving the underscore `\_' (i.e., wildcard) in query patterns in the \texttt{LIKE} predicate.
    \item \textbf{Similar protein retrieval with motif
    matching.} In biological databases, proteins and genomic sequences are increasingly represented by learned vector embeddings that capture functional or structural similarity (e.g., AlphaFold \cite{alphafold}, ProtBERT \cite{protbert}). For instance, the SwissProt~\cite{prot} dataset contains proteins with a total sequence length reaching hundreds of millions of symbols. These vector embeddings are inherently tied to their underlying sequences, where patterns such as conserved motifs or binding sites play a critical role. For instance, in Figure \ref{fig:example}, a biologist may wish to retrieve proteins that are functionally similar to a query protein (i.e., vector similarity) while containing a specific amino-acid motif (e.g., \texttt{Cys--X$_2$--Cys}) that is essential for metal binding \cite{Bailey2009-kl}. Such queries require enforcing both vector similarity and pattern constraints.
    \item \textbf{Semantic document retrieval with keyword filter.} In retrieval-augmented generation (RAG) systems \cite{rag}, documents are divided into chunks and represented as vectors for retrieval. While vector similarity captures semantic relatedness, users may impose keyword constraints (e.g., a specific person or country name) on the chunk text for fine-grained retrieval. For more details, we provide a case study of a concrete scenario for this usage in Section \ref{exp:long-sequence}.
\end{itemize}

\noindent\textbf{Challenge: Pattern constraints exhibit various selectivity and demand specialized index designs.} In constrained vector search, different filtering conditions induce varying selectivity. Existing approaches to filtered approximate nearest neighbor search can be classified into three categories: \emph{pre-filtering}, \emph{post-filtering}, and \emph{joint filtering}~\cite{attributeannsurvey}. Pre-filtering first identifies qualifying vectors using an external index on vector-associated attributes and then performs a brute-force search over the filtered subset. However, when filter selectivity is high (i.e., when a large fraction of vectors qualify), the filtered subset becomes large, leading to substantial query overhead and poor scalability. Post-filtering, in contrast, first performs ANNS on the full vector index and subsequently filters the retrieved candidates. This approach is effective when selectivity is high enough that most nearest neighbors satisfy the constraint. However, under low-selectivity conditions, a large portion of retrieved candidates are discarded, resulting in degraded recall or increased query cost. Unfortunately, \textbf{pattern constraints over sequences exhibit various selectivity}. Complex or long patterns typically match only a small number of sequences, leading to low selectivity, whereas short or simple patterns may match a large fraction of the dataset. As a result, neither pre-filtering nor post-filtering can robustly handle pattern-constrained queries across different selectivity levels. Joint filtering aims to integrate filtering conditions directly into the ANNS process, but typically requires specialized index designs tailored to specific types of constraints. Existing joint-filtering indices primarily target \emph{attribute-in} and range-based predicates over categorical or numerical attributes. These methods are designed to support operators such as equality, subset inclusion, and numerical comparison. However, such techniques cannot be directly extended to pattern-based predicates over sequence attributes. Efficiently supporting \emph{pattern constraints over sequences} within ANNS remains largely unexplored and poses a significant open challenge.

Figure~\ref{tab:comparison} shows the undesirable query trade-offs of PreFiltering and PostFiltering approaches. We conduct a case study in Figure \ref{fig:pattern} on the SIFT dataset \cite{pq} augmented with synthetic sequences of length 50. The system configuration and query generation are identical to those used in our main experiments (Section \ref{sec:exp-setup}). We observe that, as the query pattern length increases, the performance of PostFiltering drops significantly, whereas PreFiltering consistently suffers from relatively low query efficiency.

\vspace{0.5mm}
\noindent\textbf{Our solution.} Our goal is to design an index that supports \emph{ANNS with pattern constraints}, while providing query efficiency and result quality comparable to unconstrained vector search.

To optimize query efficiency and avoid unnecessary computations, a straightforward approach (OptQuery) is to query on the vector index containing only those vectors satisfying the given pattern constraint. This requires constructing the corresponding index for every possible pattern. However, the number of distinct substrings across all sequences is $O(m^2)$ in the worst case for a total sequence length $m$, leading to unaffordable space cost. While inverted file indexes \cite{bm25,ivf-knuth} reduce space consumption by avoiding materialization of all substrings, their query processing requires merging multiple posting lists that may involve duplicated vectors or false positives, leading to lower query performance. 

To overcome these limitations, we present \textsc{VectorMaton}, an automaton-based vector index for pattern-constrained ANNS. 
The key insight is that many patterns exhibit identical occurrence behavior across sequences. Such patterns can therefore be grouped and share the same index. For example, consider two sequences ``ab'', ``aab''. Then patterns ``b'' and ``ab'' are considered equivalent since they both occur at the end of each sequence. To realize this idea, we borrow the ideas from formal language theory, and leverage the \emph{suffix automaton} (SAM)~\cite{automaton,gsa}.
While SAM has traditionally been used for substring matching, it is not designed for pattern-constrained vector search. In this work, we extend this technique to both index a collection of sequences and integrate it with approximate nearest neighbor search. Since each pattern can still be directly associated with the vectors whose sequences contain it, the query efficiency remains comparable to OptQuery. Our resulting structure achieves efficient query processing while maintaining a worst-case space complexity of $O(m^{1.5})$, where $m$ is the total sequence length. We further propose two space-saving strategies to reuse indexes across states to reduce redundancy and selectively construct indexes based on pattern selectivity to adapt to different query conditions. Together, these strategies enable a unified and scalable solution for pattern-constrained ANNS.
Empirically, we observe that the index size of \textsc{VectorMaton} grows near-linear with respect to the total sequence length, while delivering substantial query performance improvements over baselines. This scalability enables \textsc{VectorMaton} to handle datasets with a total sequence length of up to hundreds of millions symbols.

\vspace{0.5mm}
\noindent\textbf{Contributions.} Our contributions are summarized as follows.
\begin{itemize}[leftmargin=*]
\item We formally define the problem of \emph{ANNS with pattern constraints} that, to the best of our knowledge, has not been systematically studied in prior work.
\item We propose \textsc{VectorMaton}, a novel index built on an enhanced suffix automaton that enables efficient pattern-constrained ANNS.
\item We conduct extensive experimental evaluations on diverse real-world datasets, demonstrating that \textsc{VectorMaton} consistently outperforms existing baselines, achieving up to two orders of magnitude higher query throughput at comparable accuracy.
\end{itemize}

\section{Preliminaries}
\subsection{Problem definition}
\label{sec:problem-definition}
Let $D=\left\{(v_1,s_1),\cdots,(v_n,s_n)\right\}$ be a dataset of size $n$, where $v_i$ denotes the $i$-th vector and $s_i$ denotes the sequence associated with $v_i$. A sequence $s$ is a string that can represent texts, protein sequence, etc. Let $m=\sum_{i=1}^n|s_i|$ be the total sequence length, and $V_p$ be the set of vectors whose associated sequences contain a pattern $p$. Pattern $p$ is also a sequence, and $s$ contains $p$ if and only if $p$ occurs as a consecutive subsequence (i.e., substring) in $s$. Given a query vector $v_q$, $V^*_{k,p}(v_q)$ is denoted as the top-$k$ vectors in $V_p$ closest to the query $v_q$. We assume $V_p\not=\emptyset$ and denote $V^*_{k,p}(v_q)=V_p$ when $k>|V_p|$. For an approximate set $V_{k,p}(v_q)$ of $V^*_{k,p}(v_q)$, we use \emph{recall} to measure its quality, defined as $Recall(R)=\frac{|R\cap R^*|}{|R^*|}$ \cite{digra}.
\begin{definition}[ANNS with pattern constraints]
\label{def:problem}
    Given a query vector $v_q$, a query pattern $p$, and an integer $k$, find a set of vectors $V_{k,p}(v_q)$ approximating the ground truth set $V^*_{k,p}(v_q)$.
\end{definition}

Our goal is to optimize the query throughput of ANNS with pattern constraints  while maintaining a high recall. We remark that while the query pattern $p$, query vector $v_q$, sequences $s_i$, and vectors $v_i$ may exhibit some degree of correlation in real-world applications, our goal is to develop a general solution whose correctness and effectiveness do not depend on such assumptions.

\begin{table}[t]
    \centering
    \caption{Summary of notations.}
    \vspace{-4mm}
    \resizebox{\linewidth}{!}{
    \begin{tabular}{ll}
        \hline
         \textbf{Notation} & \textbf{Description} \\
         \hline
         $n,m$ & the number and total length of all sequences \\
         $D$ & the dataset $\left\{(v_1,s_1),\cdots,(v_n, s_n)\right\}$ \\
         $V$ & base vectors $\left\{v_1,\cdots, v_n\right\}$ in the dataset \\
         $S$ & sequences $\left\{s_1,\cdots,s_n\right\}$ associated with the vectors \\
         $v_q,p$ & query vector $v_q$ and query pattern $p$ \\
         $V_p$ & base vectors whose associated sequences contain $p$ \\
         $V^*_{k,p}(v_q)$ & top-$k$ vectors in $V_p$ closest to the query vector $v_q$ \\
         $V_{k,p}(v_q)$ & approximated vector set of $V^*_{k,p}(v_q)$ \\
         $M$ & the maximum degree in HNSW graph \\
         $ef\_con,ef\_search$ & candidate list capacity in HNSW construction/query \\
         $T$ & threshold of constructing HNSW graph in \textsc{VectorMaton} \\
         \hline
    \end{tabular}}
    \label{tab:notation}
\end{table}
\subsection{Hierarchical Navigable Small Worlds}
The \emph{Hierarchical Navigable Small Worlds} (HNSW) graph~\cite{hnsw} is a widely adopted graph-based index for ANNS in high-dimensional vector spaces. In the following, we provide a brief overview of it.

\vspace{0.5mm}
\noindent\textbf{HNSW structure.}
HNSW organizes vectors into a multi-layer graph. The bottom layer contains all vectors and provides dense local connections among nearby vectors, while higher layers contain progressively fewer nodes and act as long-range “highways”. 
The maximum degree of a node is limited by the hyperparameter $M$. Consequently, the overall space complexity of HNSW is linear in the number of indexed vectors.


\vspace{0.5mm}
\noindent\textbf{Query processing in HNSW.}
To answer a query of top-$k$ nearest vectors, HNSW performs a greedy search starting from an entry point at the top layer, iteratively moving to neighboring nodes that are closer to the query vector until we find the approximated nearest one, and uses it as the entry node of the next layer. The trade-off between query latency and recall is controlled by the parameter \textit{ef\_search}. Similarly, the trade-off between index construction efficiency and quality is controlled by the parameter \textit{ef\_con}. Larger \textit{ef\_search} increases recall at the cost of higher query latency, while smaller \textit{ef\_search} favors faster queries with reduced accuracy.
\subsection{Suffix automaton}
\label{sec:suffix-automaton}
The \emph{suffix automaton} (SAM) \cite{automaton,BLUMER198531} is a compact deterministic finite automaton that represents all suffixes of a given sequence. The general suffix automaton (GSA) further extends SAM to handle a collection of multiple sequences. Since a \texttt{CONTAINS} predicate checks whether a pattern $p$ appears as a substring, it can be equivalently viewed as checking whether $p$ is a prefix of at least one suffix of the sequence. Therefore, both SAM and GSA naturally support efficient verification of whether a pattern is contained in a single sequence or in a sequence collection.

\vspace{0.5mm}
\noindent\textbf{Overview of SAM.}
A SAM is represented as a directed acyclic graph whose nodes, called \emph{states}, compactly represent groups of substrings with shared characteristics. Directed edges, referred to as \emph{transitions}, are labeled with single characters and indicate how substrings can be extended. SAM has a unique source node called the \emph{initial state}, from which all other states are reachable. There is a one-to-one mapping between substrings of the sequence and transition paths starting from the initial state. SAM can represent all substrings in linear space, while allowing pattern checks in time linear to the pattern length.

\begin{example}
Figure~\ref{fig:sam}(a) shows the SAM constructed for the sequence ``banana'', where state~0 is the initial state. Given the pattern $p=\text{``nana''}$, we start from state~0 and follow transitions labeled \texttt{n}, \texttt{a}, \texttt{n}, and \texttt{a}. Since the traversal successfully reaches a state in the end, pattern $p$ is contained in the sequence ``banana''.
\end{example}

\noindent\textbf{SAM construction.} The SAM of a sequence is built incrementally with one character at a time. Let $last$ denote the most recently created state corresponding to the full prefix constructed so far. When a new character $c$ arrives, we create a new state and add a transition labeled $c$ from $last$. This extends the suffixes in $last$ to end at the new position. We then follow a small chain of \emph{suffix links} \cite{automaton} from $last$ to update earlier states that should also gain a transition labeled $c$. These suffix links connect states representing substrings with shared suffixes. This incremental procedure constructs the SAM in $O(|s|)$ time for a sequence $s$.
\begin{example}
Figure~\ref{fig:sam}(b) illustrates a partial construction of the SAM for the sequence ``banana''. After processing the prefix ``ba'', state~2 is the current $last$ state and represents the suffixes ``ba'' and ``a''. When the next character `n' is added, these suffixes are extended to ``ban'' and ``an'', which are captured by a new state~3. In addition, the single-character suffix ``n'' is also introduced. To account for this new suffix, the construction follows the suffix link of state~2 and adds a transition labeled `n' from the initial state~0 to state~3.
\end{example}

\begin{figure}[t]
    \centering
    \includegraphics[width=\linewidth]{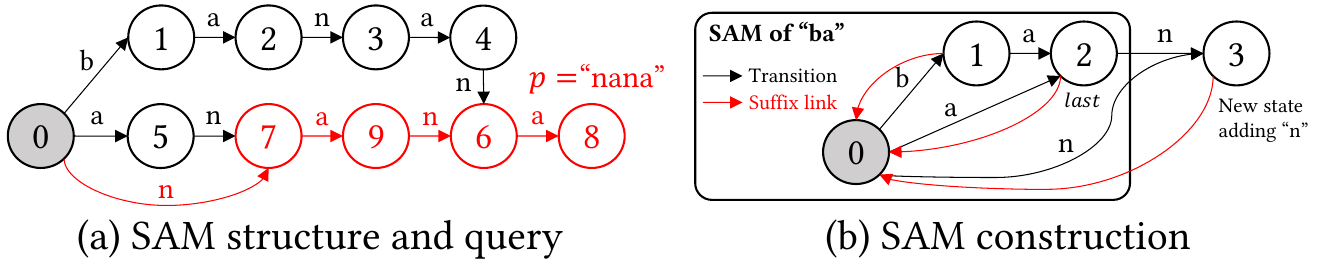}
    \vspace{-7mm}
    \caption{The SAM structure of the sequence ``banana''.}
    \label{fig:sam}
\end{figure}
\section{Proposed baselines}
\noindent\textbf{Optimal query baseline for HNSW.} We begin by proposing a fundamental question in solving the pattern-constrained ANNS problem in Definition~\ref{def:problem}:
\emph{under the HNSW query framework, what is the best achievable trade-off between query efficiency and recall?}

An intuitive solution (OptQuery) is obtained by decoupling pattern filtering from vector search. Specifically, let $V_p\subseteq V$ denote the subset of items whose associated sequences contain pattern $p$. We first identify \( V_p \), and then build and query an HNSW index over \emph{only} these qualified items. Running a standard HNSW search on this filtered graph avoids any exploration of disqualified items without affecting the accuracy. Therefore, this baseline requires enumerating all patterns that appear in the dataset and building an independent HNSW graph for each pattern $p$ over the filtered vector set $V_p$. Specifically, for each sequence $s_i$, we extract all of its substrings and insert the associated vector $v_i$ into the HNSW graph corresponding to each substring pattern. At query time, we first check whether the queried pattern exists in the index. If so, we directly perform an HNSW search on the graph associated with that pattern; otherwise, no feasible result exists and the query terminates immediately. As each sequence $s_i$ contains $O(|s_i|^2)$ substrings, the total space required for OptQuery is bounded by $O(\sum_{i} |s_i|^2)=O(m^2)$.

\begin{figure}[t]
    \centering
    \includegraphics[width=1.\linewidth]{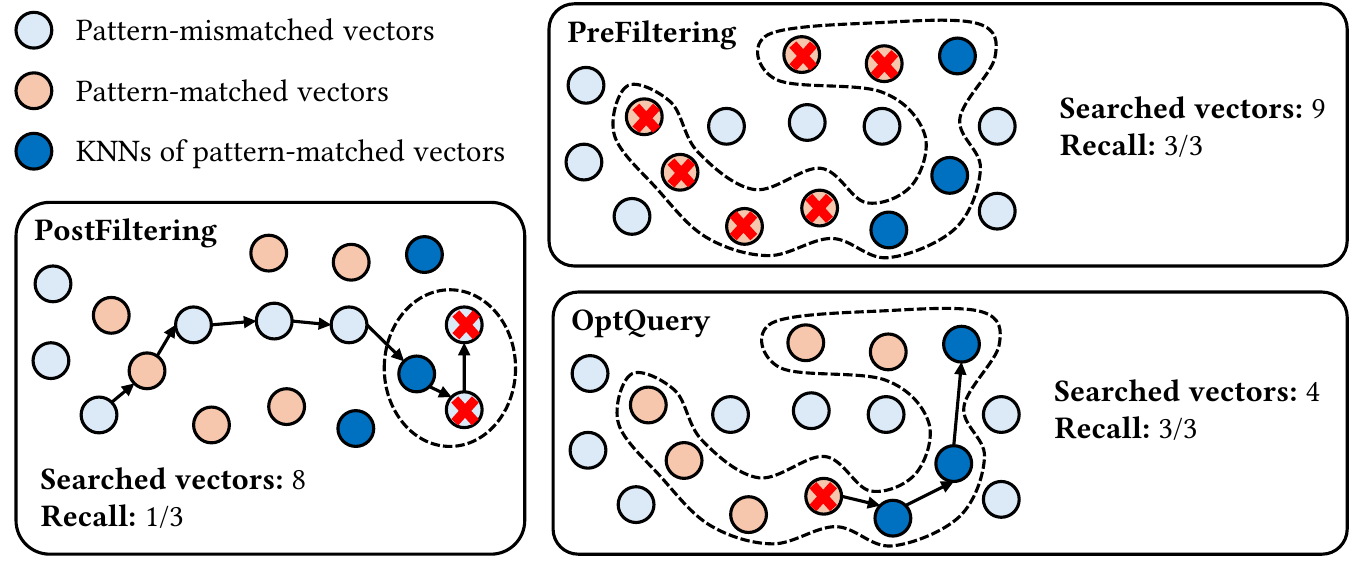}
    \vspace{-8mm}
    \caption{Comparison of different baselines.}
    \label{fig:overview}
\end{figure}








\vspace{0.5mm}
\noindent\textbf{PreFiltering and PostFiltering.} Algorithm \ref{alg:filtering} shows the query processes of \textit{PreFiltering} and \textit{PostFiltering}. Given the query vector $v_q$, query pattern $p$ and integer $k$, PreFiltering first identifies candidate vectors $V_p$ and then performs a brute-force search over this filtered subset to find $k$ nearest neighbors to $v_q$. This approach guarantees exact results but often incurs significantly higher query time due to the exhaustive search. PostFiltering, in contrast, performs a standard ANN search over the full HNSW index and subsequently filters the results to retain only those satisfying the pattern constraint. With the same search-time budget as OptQuery, it may sacrifice recall because some relevant vectors can be missed during the initial ANN search.

We further introduce a \textit{Hybrid} strategy that dynamically switches between PreFiltering and PostFiltering according to the query selectivity. Specifically, since PostFiltering retains only the satisfying vectors among the top-$ef\_search$ ANN results, the expected number of valid results is approximately $ef\_search \cdot s$, where $s$ denotes the selectivity (i.e., the fraction of sequences satisfying the query pattern). Therefore, we use $\frac{k}{ef\_search}$ as the selectivity threshold. When $s < \frac{k}{ef\_search}$, the expected number of satisfying results returned by PostFiltering is less than $ef\_search \cdot \frac{k}{ef\_search} = k$, which can significantly degrade recall. In this case, we apply PreFiltering to guarantee sufficient candidate vectors. Otherwise, we apply PostFiltering to achieve higher query efficiency.
\begin{algorithm}[ht]
\small
\caption{Pre/Post-Filtering}
\label{alg:filtering}

\SetKwProg{Fn}{Function}{:}{}
\SetKwFunction{PreFiltering}{PreFiltering}
\SetKwFunction{PostFiltering}{PostFiltering}

\Fn{\PreFiltering{$v_q, p, k$}}{
    Identify $V_p$ from an index of $S$\;
    Return $k$ closest vectors in $V_p$\;
}

\BlankLine

\Fn{\PostFiltering{$v_q, p, k, ef\_search$}}{
    Search the global HNSW graph over $V$ with parameters $v_q,k,ef\_search$ and return $ef\_search$ results\;
    Return $k$ closest vectors whose associated sequences contain $p$\;
}

\end{algorithm}

\vspace{0.5mm}
\noindent\textbf{Inverted-index-based filtering.} Inverted indexes (e.g., BM25~\cite{bm25}) are widely used for textual retrieval and can be adapted as an approximate filtering mechanism for pattern-constrained vector search. Therefore, we consider an additional approach, denoted as \textit{BM25Filtering}, which combines an inverted index with HNSW search. Specifically, BM25Filtering first builds an inverted index over all 2-grams appearing in the sequence collection. For each 2-gram, an HNSW graph is constructed over the vectors whose associated sequences contain that 2-gram. Given a query pattern $p$, the algorithm decomposes $p$ into its constituent 2-grams and queries the corresponding HNSW graphs to retrieve candidate vectors. The retrieved candidates are then ranked according to their BM25 scores with respect to $p$, and the top-ranked candidates are further sorted by vector distance to produce the final top-$k$ results. However, we would like to remark that BM25 (or other inverted index methods) is designed for textual relevance rather than exact pattern matching. Consequently, sequences with high BM25 scores may not actually contain the query pattern. Even if a sequence contains the query pattern, it may still be filtered out during the BM25 ranking process as it contains less 2-grams than other candidate sequences. Therefore, some results returned by BM25Filtering may not satisfy the pattern constraint. Moreover, each queried HNSW graph contains vectors associated with a particular 2-gram rather than the complete pattern, introducing a large number of false-positive candidates that harm the query performance.

\vspace{0.5mm}
\noindent\textbf{Open challenges and opportunities.} As shown in Figure \ref{fig:overview}, PreFiltering first identifies all pattern-matched vectors and then performs vector search over the filtered subset. This leads to examining a large number of candidates and consequently low query efficiency. PostFiltering, on the other hand, performs ANN search before applying the pattern constraint. As a result, many retrieved candidates may not satisfy the pattern, which degrades recall or requires expanding the search space to compensate. OptQuery achieves the best efficiency-recall trade-off by constructing vector indexes for pattern-matched subsets, ensuring that only relevant vectors are examined during search. However, this strategy incurs a quadratic space cost, as it potentially requires building indexes for $O(m^2)$ distinct patterns, making it impractical for large-scale datasets.

The key challenge is therefore how to retain the query performance of OptQuery while reducing its index size. Our central insight is that although there can be $O(m^2)$ possible substrings in a sequence collection, many of them share identical occurrence behavior and can be grouped into $O(m)$ equivalence classes via an automaton structure. By constructing indexes over these equivalence classes rather than individual patterns, we can significantly reduce the index size while preserving efficient query processing.

\begin{figure}[t]
    \centering
    \includegraphics[width=\linewidth]{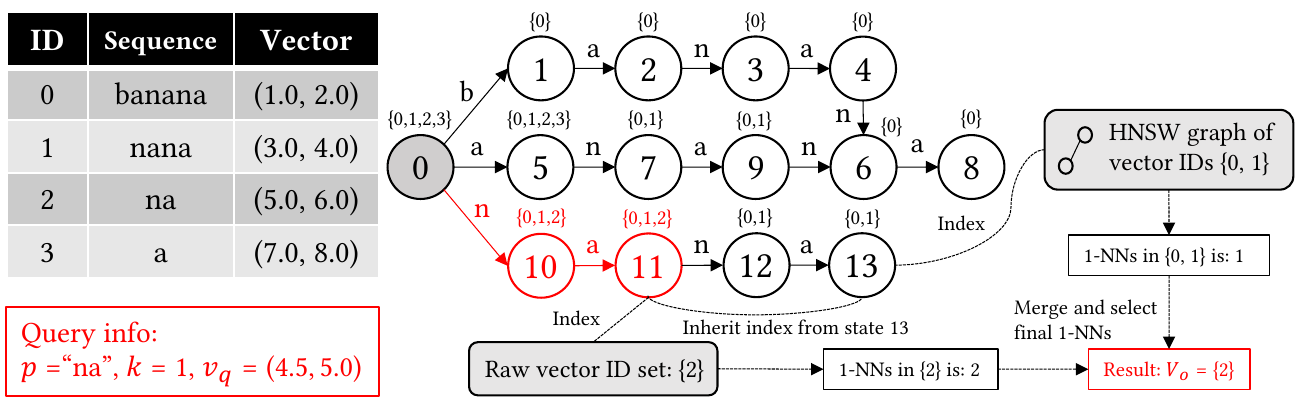}
    \vspace{-8mm}
    \caption{An example \textsc{VectorMaton} index and query.}
    \label{fig:vectormaton}
\end{figure}
\section{\textsc{VectorMaton}}
We now present \textsc{VectorMaton}, an index designed to support efficient approximate nearest neighbor search under pattern constraints. We first describe the overall index structure in Section~\ref{sec:index-structure}. As illustrated in Figure \ref{fig:vectormaton}, \textsc{VectorMaton} consists of an enhanced suffix automaton (ESAM) that indexes all substrings of the sequences. To process a query pattern $p$, we traverse ESAM along edges labeled with the symbols of $p$ until reaching the corresponding state. From this state, we perform searches over its associated vectors. The results from these searches are then merged to produce the final top-$k$ nearest neighbors. We then detail the query processing and index construction procedures in Sections~\ref{sec:query-processing} and ~\ref{sec:index-construction}, respectively, and the index maintenance scheme in Section \ref{sec:maintenance}.

\subsection{Index structure}
\label{sec:index-structure}
Figure~\ref{fig:vectormaton} illustrates an index overview of \textsc{VectorMaton}. We propose an enhanced suffix automaton, termed as ESAM, where each state is associated with a corresponding state index either as an HNSW graph or as a raw vector ID set.

\vspace{0.5mm}
\noindent\textbf{ESAM states.} The proposed \textsc{VectorMaton} index constructs an enhanced suffix automaton over all sequences in the collection. Recall that the general suffix automaton (GSA)~\cite{gsa} compactly represents all substrings of a sequence collection and supports efficient \emph{pattern existence} queries (see Section~\ref{sec:suffix-automaton}). However, GSA cannot be directly extended to determine \emph{which sequences} contain a given pattern. This limitation arises because a single automaton state may represent multiple substrings whose occurrences span different subsets of sequences (see Example \ref{example:gsa}).
\begin{figure}[t]
    \centering
    \includegraphics[width=.6\linewidth]{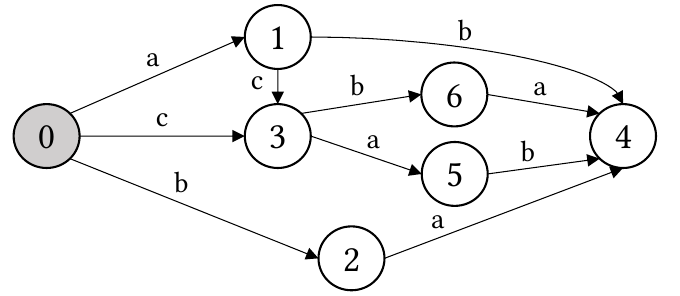}
    \vspace{-4mm}
    \caption{The GSA \cite{gsa} of sequences \texttt{ac}, \texttt{acab} and \texttt{acba}.}
    \label{fig:gsa}
\end{figure}
\begin{example}
    \label{example:gsa}
    Figure~\ref{fig:gsa} illustrates a GSA built over the sequences \texttt{ac}, \texttt{acab}, and \texttt{acba}. Consider two patterns, $p_1=\texttt{cba}$ and $p_2=\texttt{cab}$. Although the GSA correctly determines the existence of both patterns by ending in the same state (state~4), their occurrence sets differ: $p_1$ occurs only in sequence \texttt{acba}, whereas $p_2$ occurs only in sequence \texttt{acab}. Consequently, the GSA cannot distinguish which sequences contain a given pattern.
\end{example}

To support pattern-constrained queries, we require each ESAM state to represent a set of \textit{equivalent patterns} that correspond to the same set of satisfying vectors. Moreover, extending any patterns within the same state by the same symbol should also correspond to the same satisfying-vector set. Similar notions of equivalence have been studied in prior work \cite{ivf}. Building on these ideas, we formalize the pattern equivalence used by ESAM.

\begin{definition}[Position list (poslist)]
Given a pattern $p$ that appears in the sequence collection, we define its \emph{position list}, denoted as $\mathit{poslist}(p)$, as the set of pairs $(\mathit{id}, \mathit{pos})$, where $\mathit{id}$ is the identifier of a sequence containing $p$, and $\mathit{pos}$ is the end position of an occurrence of $p$ in that sequence.
\end{definition}

\begin{example}
Consider the pattern $p =$ \texttt{anan}, which corresponds to state~6 in Figure~\ref{fig:vectormaton}. The pattern occurs only once, in the sequence \texttt{banana} with sequence ID~1, ending at position~4 of the sequence. Therefore, $\mathit{poslist}(p) = \left\{(1,4)\right\}$.
\end{example}

\begin{definition}[Equivalence class]
Two patterns $p_1$ and $p_2$ are considered equivalent if and only if $\mathit{poslist}(p_1) = \mathit{poslist}(p_2)$.
\end{definition}

\begin{definition}[Maximal pattern]
    Given a pattern $p$ that appears in the sequence collection, we say $p$ is maximal if and only if $p$ is longest within its equivalence class.
\end{definition}


As a result, the patterns of each equivalence class correspond to only one set of vector IDs, which is the union of IDs in its position list. Moreover, it is straightforward to observe that each equivalence class contains exactly one maximal pattern. In particular, if two maximal patterns have the same length and are equivalent, their occurrences in the sequences are thus identical, and extracting the pattern from their occurrences yields the same resulting pattern.

Since each state in ESAM represents a set of equivalent patterns, the number of states is equal to the number of maximal patterns in the sequence collection. Consequently, bounding the number of maximal patterns directly yields a bound on the number of states. Due to space limit, for all following lemmas and theorems, we defer their proofs to the appendix of our extended version \cite{extended}.

\begin{lemma}
\label{lemma:num-states}
Let $S$ be a collection of sequences with total length $m$. The number of maximal patterns in $S$ is bounded by $O(m)$. Therefore, the number of states in ESAM is also bounded by $O(m)$.
\end{lemma}


The linear bound on the number of states is a fundamental property of suffix automata. Lemma~\ref{lemma:num-states} demonstrates that our ESAM preserves this property, avoiding the potential for substantial space blowup. In the following, we further establish an upper bound on the total size of the ID sets associated with all states.

\begin{lemma}
\label{lemma:num-id-set}
The total size of all associated ID sets in ESAM is bounded by $O(m^{1.5})$, where $m$ is the total length of all sequences.
\end{lemma}



\noindent\textbf{ESAM transitions.}
As we have defined the states of ESAM, we next describe how transitions between states are established. As discussed in Section~\ref{sec:suffix-automaton}, transitions in classical suffix automata correspond to extending a pattern by one character. This principle naturally generalizes to our multi-sequence setting. Specifically, the following lemma formalizes the correctness of the resulting transition construction.

\begin{lemma}
\label{lemma:transition}
Let $A$ and $B$ be two equivalence classes, corresponding to states $A$ and $B$. If there exist patterns $p_A \in A$ and $p_B \in B$ such that $p_B = p_A \cdot c$ for some character $c$, then for every other pattern $p_A' \in A$, the extended pattern $p_A' \cdot c$ also belongs to $B$. Consequently, a transition labeled $c$ can be established from state $A$ to state $B$.
\end{lemma}

Since each transition strictly increases all $pos$ values by one, the states and transitions in ESAM form a directed acyclic graph (DAG). Moreover, for any transition $i \to j$, the associated vector ID sets satisfy $V_j \subseteq V_i$. This monotonicity property enables effective reuse of index information between a state and its descendant states.

\vspace{0.5mm}
\noindent\textbf{State index.}
While each state in ESAM is associated with a vector ID set, as illustrated in Figure~\ref{fig:vectormaton}, these sets exhibit substantial overlap across states. Moreover, different states correspond to patterns with varying selectivity. Consequently, constructing an independent HNSW index for every state would be highly space-inefficient. To address this issue, we propose two space-optimization strategies: \textbf{(1) an index reuse strategy} and \textbf{(2) a skip-build strategy}.

For index reuse strategy, let $V_j$ be the associated ID set over state $j$, and $I_j$ be the ID set to index associated with state $j$ that contains $|I_j|$ vectors. Among all descendants of $j$ (i.e., states reachable from $j$ via a directed path), let $k$ be the descendant with the largest indexed ID set. We call state $k$ the \emph{inherited state} of state $j$. Instead of building $I_j$ over the entire set $V_j$, we construct $I_j$ only over the difference set $V_j \setminus I_k$. During query processing on state $j$, we independently retrieve the $k$-NNs from both $I_j$ and $I_k$, and then merge the two result sets. This reuse strategy achieves comparable query accuracy as an index built directly over $V_j$ while reducing the total space consumption.

For skip-build strategy, as states may have widely varying ID set sizes due to differences in pattern selectivity, constructing a full HNSW graph for every state is often unnecessary. Under the skip-build strategy, when the size of a state’s ID set is below a threshold $T$, we store only the raw vector ID set and perform brute-force search when querying. This approach is generally more efficient than graph-based search for very small datasets. For example, when the search parameter $ef\_search$ exceeds the graph size, HNSW search may degenerate into an exhaustive traversal of all vertices, which is typically slower than direct brute-force evaluation.

\subsection{Query processing}
\label{sec:query-processing}
To answer a query, we first leverage the ESAM to locate the state corresponding to the query pattern, which is associated with the filtered vector set. We then perform $k$-NN search over the index associated with this state and, if applicable, its inherited index. Algorithm~\ref{alg:vectormaton} presents the details. The algorithm begins at the initial state $0$, sequentially scans the query pattern and follows the corresponding transitions in ESAM until either the entire pattern is consumed or no valid transition exists (lines 26-29). The algorithm then retrieves the index associated with final state as well as the index of its inherited state if it exists (lines 30-31). Finally, it independently performs $k$-nearest neighbor searches on the two indexes and merges the results to obtain the final answer set (lines 32-33).

As \textsc{VectorMaton} adopts a skip-build strategy that adaptively constructs each index either as an HNSW graph or as a raw vector ID set, the query processing is different on the two types of index. If an index is implemented as an HNSW graph, we apply the standard HNSW search procedure with parameter $ef\_search$. Otherwise, we perform a brute-force scan over all associated vectors and select the top-$k$ nearest neighbors. Moreover, the indexes of a state and its inherited state together form an exact cover of the state’s associated vector set, and therefore merging the top-$k$ results returned from these two indexes is \textbf{lossless}. Consequently, the resulting recall is comparable to that of OptQuery.

\begin{algorithm}[ht]
\small
\caption{\textsc{VectorMaton}}
\label{alg:vectormaton}

\SetKwProg{Fn}{Function}{:}{}
\SetKwFunction{Build}{Build}
\SetKwFunction{Query}{Query}

\Fn{\Build{$V,S,T$}}{
  \tcp{Build ESAM via suffix-link extension}
  Initialize ESAM with state $0$\;
  \ForEach{sequence $s_i \in S$}{
    $last \gets 0$\;
    \ForEach{symbol $c$ in $s_i$}{
      \lIf{$last$ has no transition $c$}{Create new state $cur$}
      \lElse{$cur\gets$ the next state of $last$ via $c$}
      Follow suffix links from $last$ and:
        (i) connect missing $c$-transitions to $cur$;
        (ii) stop at first state $A$ with $c$-transition\;
      \lIf{no such state $A$ exists}{
        $\mathrm{link}(cur) \gets 0$
      }\Else{
        $B\gets$ the $c$-successor of $A$\;
        \If{$|B|=|A|+1$}{
          \lIf{$cur\not=B$}{$\mathrm{link}(cur) \gets B$}
        }\Else{
          Clone $B$ as $B'$ and redirect transitions\;
          \lIf{$cur\not=B$}{$\mathrm{link}(cur) \gets B'$}\lElse{$cur\gets B'$}
        }
      }
      $last \gets cur$\;
      \tcp{ID propagation} 
      Follow suffix link from $last$ and (1) add ID $i$ to their ID set; (2) stop at first state that already contains $i$\;
    }
  }

  \tcp{Build state indexes}
  \ForEach{state $u$ in reverse topological order}{
    $v\gets$ descendant with largest indexed set\;
    $base(u)\gets ID(u)$ excluding the indexed set of $v$\; 
    \lIf{$|base(u)|<T$}{
        Store $I(u)\gets$ raw ID set of $base(u)$
    }\lElse{
        Build index $I(u)\gets$ HNSW over $base(u)$
    }
  }
}
\BlankLine

\Fn{\Query{$v_q, p, k, ef\_search$}}{
    $cur\gets$ the state 0\;
    
    \For{$i\gets 0, 1,\cdots,|p|-1$}{
        \lIf{no transition $p_i$ in $cur$}{
            \textbf{return} $\emptyset$
        }
        $cur\gets$ next state of transition $p_i$\;
    }
    $next\gets$ the state that $cur$ inherits index from\;
    $I(cur),I(next)\gets$ the index of $cur,next$\;
    $res_1,res_2\gets$ the $k$-NN results of $v_q$ from $I(cur),I(next)$\;
    \textbf{return} the top-$k$ closest vector IDs in $res_1\cup res_2$\;
}
\end{algorithm}
\begin{lemma}
    \label{lemma:exact-cover}
    Let $V_j$ be the associated ID set over state $j$, and $I_j$ be the ID set of the index (HNSW or raw ID set) of state $j$. Then for each state $j$ and its inherited state $k$, the sets $I_j,I_k$ form an exact cover of $V_j$, i.e., $I_j\cup I_k=V_j$ and $I_j\cap I_k=\emptyset$.
\end{lemma}

\subsection{Index construction}
\label{sec:index-construction}
In this section, we introduce how to construct the ESAM and state index, respectively. Additionally, we provide a parallel construction approach that accelerates the index construction process.

\vspace{0.5mm}
\noindent\textbf{ESAM construction.} For ESAM construction, we propose a generalized variant of the SAM construction described in Section~\ref{sec:suffix-automaton}. We process the sequence collection incrementally, and for each sequence, we process its symbols in order. Next, we aim to identify new states and transitions upon processing each symbol.

The key question is \textbf{which old states may give rise to new states}. Since symbols are processed incrementally, each newly read symbol extends all suffixes of the currently processed sequence by one character, potentially forming new equivalence classes. Consequently, the construction must maintain the states that correspond to suffixes of the current sequence.

\begin{figure}[t]
    \centering
    \includegraphics[width=\linewidth]{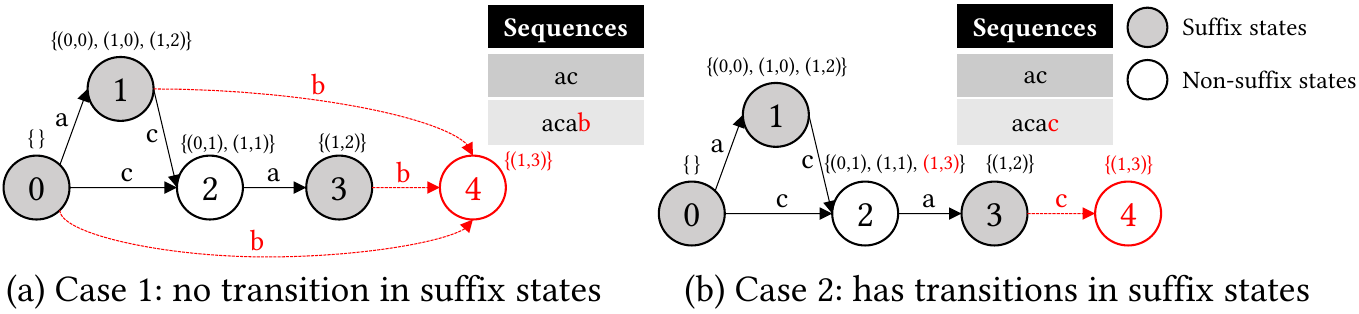}
    \vspace{-7mm}
    \caption{An example of suffix states and their extension.}
    \label{fig:suffix-states}
\end{figure}
\begin{definition}[Suffix state]
Consider the sequence $s_i$ currently being processed and suppose the algorithm has processed position $j$. Upon reading the next symbol (at position $j+1$), a state (i.e., an equivalence class) is called a \emph{suffix state} if and only if (1) its position list is empty, or (2) its position list contains $(i,j)$.
\end{definition}

The next question is \textbf{how to extend these suffix states to new states}. By Lemma~\ref{lemma:transition}, for a suffix state $A$ with position list $poslist(A)=\left\{(x_k,y_k)\right\}$, extending $A$ by the next symbol $c=s_i[j+1]$ yields a state whose position list is $\left\{(x_k,y_k+1)|s_{x_k}[y_k+1]=c\right\}$. Consequently, if a suffix state $A$ has no outgoing transition labeled $c$ prior to the extension, none of its existing occurrences can be extended by $c$, and the extension introduces only the new occurrence ending at $(i,j+1)$. Therefore, all such suffix states point to the same newly created state whose position list is exactly $\left\{(i,j+1)\right\}$.

\begin{lemma}
\label{lemma:case1}
Let $c=s_i[j+1]$ be the next symbol to be processed. For all suffix states $A$ that do not have an outgoing transition labeled $c$, extending $A$ by $c$ leads to the same state.
\end{lemma}

Conversely, suppose a suffix state $A$ already has an outgoing transition labeled $c$ prior to the extension, and let $B$ denote its target state. We examine whether all occurrences represented by $B$ remain valid after introducing the new occurrence $(i,j+1)$. If so, then no structural change is required: $B$ simply augments its position list with $(i,j+1)$ and continues to represent the updated class. Otherwise, the extension violates the equivalence of occurrences in $B$. In this case, we split the class by creating a copy $B'$, which includes all outgoing transitions from $B$. The original state $B$ retains the occurrences that cannot be extended to $(i,j+1)$, while $B'$ represents the refined equivalence class that includes $(i,j+1)$, and then replace connection $A\to B$ by $A\to B'$.

\begin{example}
Figure \ref{fig:suffix-states} shows examples of extending suffix states. In (a), when suffix states have no corresponding transition `b'. We create a new state and connect these suffix states to the new state with position list $\left\{(1,3)\right\}$. In (b), the suffix states 0 and 1 both have an outgoing transition labeled `c'. Since all occurrences represented by state 2 remain consistent after adding the new occurrence $(1,3)$, no structural change is required.
\end{example}

For ease of implementation, we utilize the \emph{suffix link} mechanism~\cite{automaton} to provide an efficient way to process them. While the concept is primarily designed for the single-sequence setting (see Section \ref{sec:suffix-automaton}), it can be naturally generalized to our multi-sequence ESAM. The detailed correctness proof and suffix link maintenance are deferred to the appendix of the extended version \cite{extended}. In the following, we outline the ESAM construction process via suffix links (Algorithm \ref{alg:vectormaton}). Specifically, the suffix states form a chain ordered by decreasing maximal pattern length via suffix links, which splits into at most two contiguous parts: states without an outgoing transition labeled $c$, and states with such a transition. In Algorithm \ref{alg:vectormaton}, $\mathrm{link}(X)$ represents the next suffix state of $X$ in its suffix link. For states without transition $c$, we create a new state and connect all of them to it via transitions labeled $c$ (line 8). For states with transition $c$, let $A$ be the first such state and let $B$ be its $c$-successor (line 11). If the maximal pattern length of $B$ equals that of $A$ plus one, then $B$ already represents the correct equivalence class after extension and no structural change is required (lines 12-13). Otherwise, we create a copy $B'$ of $B$ so that $B'$ represents position list extendable by the new symbol (i.e., $poslist(B')=poslist(B)\cup\left\{(i,j+1)\right\}$), while $B$ preserves the original position list. We then update the remaining suffix states on the chain by redirecting their $c$-transitions from $B$ to $B'$ (lines 14-15). After processing a sequence, we need to reset suffix state to 0 before processing the next sequence. 

\vspace{0.5mm}
\noindent\textbf{Vector ID propagation.}
Vector IDs can be propagated in an online manner during ESAM construction. Since each extension step processes all suffix states along the suffix-link chain, these states correspond exactly to patterns that end at the current position and should therefore include the current sequence ID. Thus, we propagate the vector ID by traversing the suffix-link chain starting from the newly created state and adding the ID until reaching a state that already contains it (line~19 of Algorithm~\ref{alg:vectormaton}).

\vspace{0.5mm}
\noindent\textbf{State index construction.} To address the significant overlap among ID sets and handle small ID sets, we propose two strategies: an \textbf{index reuse strategy} and a \textbf{skip-build strategy}. Specifically, we construct state indexes in reverse topological order of ESAM. In this order, when processing a state $u$, all its descendant states have already been processed. We identify the descendant of $u$ with the largest indexed ID set in constant time. We then define the \emph{base set} of $u$ as the difference set between $ID(u)$ and the largest indexed set of its descendants. If $|base(u)| < T$, we directly store $base(u)$ as a raw ID set. Otherwise, we build an HNSW index over the vectors corresponding to $base(u)$. This process ensures effective index reuse while avoiding unnecessary graph construction for small sets. The detailed procedure is shown in lines 20–24 of Algorithm~\ref{alg:vectormaton}.

We remark that all vectors are stored in a global array. Each HNSW graph maintains only the IDs of vectors rather than storing local copies of the vectors themselves. Therefore, the total space usage of \textsc{VectorMaton} is dominated by the storage of ID sets and ESAM structures, and is linear in the total size of all ID sets.

\begin{theorem}
\label{theorem:space-complexity}
Let $m$ denote the total length of all sequences. With a constant size of symbol set, the overall space complexity of the \textsc{VectorMaton} index is bounded by $O(m^{1.5})$.
\end{theorem}

\noindent\textbf{Parallel construction.} The primary bottleneck in index construction is building the HNSW graphs. Although the construction of a single HNSW graph is inherently sequential and difficult to parallelize efficiently, different state indexes are independent once their descendant states have been processed. To exploit this property, in parallel settings, we maintain a concurrent queue that stores states whose descendant states have all been processed. Worker threads repeatedly dequeue these states, construct their corresponding indexes (either raw sets or HNSW graphs), and enqueue their predecessor states if appropriate. This design enables parallel construction of multiple HNSW graphs to improve construction efficiency while preserving correctness. Our experimental results show that with 16 threads, the index can be constructed within 2 hours on datasets with hundreds of millions sequence length.

\subsection{Index maintenance}
\label{sec:maintenance}
\noindent\textbf{Insertion.} Since ESAM is constructed incrementally, it can naturally supports insertion. The main challenge lies in maintaining the index reuse strategy while preserving the correctness of Lemma~\ref{lemma:exact-cover}. Failure to maintain this property may result in increased index size and degraded query performance. During static construction, the exact-cover property is established by a reverse topological traversal after all insertions. Repeating this traversal after every dynamic insertion would be prohibitively expensive. We therefore propose an online update scheme that preserves Lemma~\ref{lemma:exact-cover} without additional asymptotic overhead. Specifically, when inserting a new sequence-vector pair, we first identify all affected states that should contain the new data, and any newly created states generated during the ESAM update process. We then update the corresponding indexes of all these states according to the following rules:

\begin{itemize}[leftmargin=*]
\item For a newly created state, we construct a local index over all vectors associated with that state directly.
\item For an existing affected state without an inherited index, we directly insert the new vector into its local index.
\item For an existing affected state with an inherited index, we recursively process its inherited state, until (1) it does not have an inherited index, or (2) its next inherited state is processed before or not affected. Then, starting from the final state reversely, we check whether the new vector is already inserted into its inherited index. If so, no update to the local index is required; otherwise, the vector is inserted into the local HNSW index.
\end{itemize}

This procedure does not introduce additional asymptotic overhead beyond the ESAM update process, since each affected state and newly created state is processed exactly once. For space limit, the proof of preserving correctness of Lemma \ref{lemma:exact-cover} is deferred to the appendix of our extended version \cite{extended}.

\vspace{0.5mm}
\noindent\textbf{Deletion.} Deletion is more challenging, as both ESAM and the HNSW index do not naturally support efficient structural removal. Instead of modifying the structure, we discuss a lazy-deletion strategy to mark the deleted data and ignore them during query processing, which is further described in the appendix of our extended version \cite{extended}. We remark that designing an efficient garbage collection mechanism to reclaim space from deleted entries remains challenging and is a direction for future work.
\section{Discussions}
\label{sec:discussions}
The pattern constraint considered in this paper primarily corresponds to substring matching, as defined in Section~\ref{sec:problem-definition}. Nevertheless, \textsc{VectorMaton} can be simply extended to support several more expressive pattern predicates.

\vspace{0.5mm}
\noindent\textbf{The underscore `\_' in the SQL \texttt{LIKE} predicate.} We first consider substring matching with the wildcard symbol `\_', which in SQL \texttt{LIKE} predicates represents an arbitrary single symbol. To support such queries, when \textsc{VectorMaton} encounters a `\_' during pattern traversal, it explores all outgoing transitions from the current state instead of following a single transition corresponding to a specific symbol. Each transition corresponds to an independent search branch. After processing the remaining pattern symbols, the results from all branches are merged and ranked according to vector distance, and the top-$k$ closest vectors are returned.

\begin{example}
In Figure~\ref{fig:vectormaton}, consider the pattern $p=$``\_a''. Starting from state 0 and `\_', all three outgoing transitions are explored, reaching states 1, 5, and 10, respectively. The second symbol is `a'. Since state 5 has no outgoing transition labeled `a', this branch terminates without producing any result. In contrast, the branches through states 1 and 10 continue to states 2 and 11, respectively, where the corresponding ANN searches are performed. Finally, the results from the two valid branches are merged, and the top-$k$ closest vectors are returned.
\end{example}

We note that the number of outgoing transitions from each state is bounded by the alphabet size, which is typically a small constant. Moreover, not all traversal branches necessarily reach a valid state, as some may terminate early due to missing transitions. Therefore, when the query pattern contains only a small number of wildcard symbols `\_', the number of additional HNSW searches remains limited and it incurs only a modest query overhead.

\vspace{0.5mm}
\noindent\textbf{Disjunctive OR queries.} For a disjunctive query containing multiple patterns, \textsc{VectorMaton} executes the query process independently for each pattern, merges the resulting candidate sets, and returns the top-$k$ vectors ranked by vector distance. Since the individual pattern queries are independent, they can be processed in parallel. We remark that \textsc{VectorMaton} achieves high query efficiency for individual pattern queries. Consequently, this will not introduce a large additional overhead when the number of disjunctive patterns is modest.
\begin{table}[t]
    \caption{Datasets used in our experiments.}
    \vspace{-4mm}
    \setlength{\tabcolsep}{3.5mm}{
    \small
    \begin{tabular}{l|c|c|c}
    \hline
    \textbf{Datasets} & \textbf{No. vectors} & \textbf{Total seq. len.} & \textbf{Dim.}\\
    \hline\hline
    spam \cite{spam} & 489 & 13,643 & 384\\
    \hline
    words \cite{efarrall2024wordembeddings} & 8,000 & 56,209 & 3,072\\
    \hline
    mtg \cite{mtg} & 21,550 & 1,504,633 & 1,152\\
    \hline
    arxiv \cite{arxiv} & 157,605 & 9,851,413 & 768\\
    \hline
    prot \cite{prot} & 455,692 & 116,326,099 & 1,024\\
    \hline
    code \cite{code} & 1,838,414 & 40,940,167 & 768\\
    \hline
    \end{tabular}}
    \label{tab:datasets}
\end{table}
\section{Experiments}
\label{lab:exp}

\subsection{Setup}
\label{sec:exp-setup}
\noindent\textbf{Hardware and datasets.} All experiments were conducted on a machine equipped with an AMD Ryzen Threadripper PRO 7975WX 32-core processor and 755 GB of RAM. The code was compiled using GCC 11.4.0 on Ubuntu Linux with the \texttt{-O3} optimization flag enabled. We evaluate all methods on a diverse set of datasets collected from Hugging Face and the SpamAssassin corpus, as summarized in Table~\ref{tab:datasets}. For space limit, detailed description of each dataset is deferred to the appendix of our extended version \cite{extended}.

\vspace{0.5mm}
\noindent\textbf{Baselines.} We compare \textsc{VectorMaton} against \textbf{OptQuery}, \textbf{PreFiltering}, \textbf{PostFiltering}, \textbf{Hybrid}, \textbf{BM25Filtering}, \textbf{ACORN-1} \cite{acorn}, \textbf{ACORN-$\gamma$} \cite{acorn}, \textbf{pgvector} \cite{pgvector} and \textbf{Elasticsearch} \cite{elasticsearch-vector}. For PreFiltering, we first use the enhanced suffix automaton as the filtering index $S$ to retrieve candidate vectors, and then perform vector search over the filtered subset. Hybrid is the strategy combining PreFiltering and PostFiltering. ACORN-1 and ACORN-$\gamma$ are general solutions for constrained vector search, and we incorporate our enhanced suffix automaton to build their filter map (which is required in query processing). Pgvector is a vector extension of PostgreSQL, which supports hybrid SQL queries with vector similarity based on HNSW index and \texttt{LIKE} predicates. To minimize disk I/O overhead and ensure fair comparison, we store all data in temporary tables and configure \texttt{temp\_buffers} to 128\,GB so that query processing is performed entirely in memory. Elasticsearch is an engine that supports hybrid queries by combining vector similarity based on HNSW index with keyword-based filtering (e.g., \texttt{match} or \texttt{wildcard} predicates). To minimize disk I/O overhead and ensure fair comparison, we allocate sufficient JVM heap space to keep all vectors and indices in memory. We remark that pgvector and Elasticsearch are implemented as a  database system, which may introduce additional overhead compared to our native implementations over the queries.


\vspace{0.5mm}
\noindent\textbf{Parameter settings.} For \textsc{VectorMaton}, we vary the HNSW search parameter $ef\_search$ in the range $[8, 1024]$, and the threshold $T$ of constructing a HNSW graph is set to 800. The maximum supported $ef\_search$ of pgvector is 1000 so we configure its range as $[8,1000]$. For other methods, we apply a larger range for $ef\_search$ as $[8,4096]$ to improve their recall. Elasticsearch does not support adjusting $ef\_search$, and we adjust its parameter ``num\_candidates'' alternatively. The maximum degree of the HNSW graph is set to $M=16$, and the construction parameter $ef\_con$ is set to 200. ACORN requires larger $M$ to ensure connectivity in its graph structure, and hence its parameter $M$ is set as 32.

\vspace{0.5mm}
\noindent\textbf{Queries.} We generate 1,000 queries for each pattern length $|p|\in\left\{2,3,4,5,6,7\right\}$ and set $k=10$ (i.e., querying 10-NNs). The query patterns are randomly sampled from substrings that occur in the sequence collection. To better preserve the empirical distribution of the vector dataset, half of the query vectors are sampled directly from the original vector set, while the remaining half are generated by applying minor perturbations to randomly sampled vectors.

\vspace{0.5mm}
\noindent\textbf{Remarks.} OptQuery encounters out-of-memory (OOM) errors on the MTG, ArXiv, SwissProt, and CodeSearchNet datasets. Pgvector does not support constructing HNSW indexes for vectors with dimensionality larger than 2000 and is thus excluded from the Words dataset. PreFiltering is omitted in cases where its QPS is less than 10\% of the lowest QPS achieved by the other methods. BM25Filtering exhibits a recall of only 0.01 on SwissProt dataset ($|p|=4$ case) and is thereby omitted in the figure.


\begin{figure*}[t]
    \includegraphics[width=1.\linewidth]{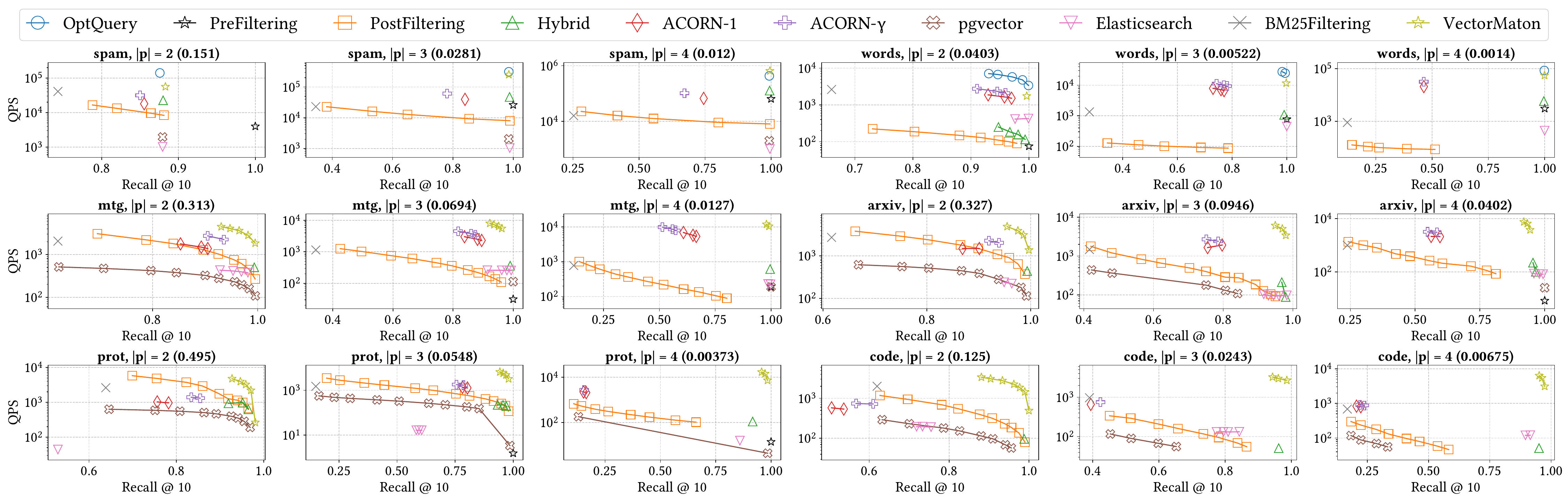}
    \vspace{-8mm}
    \caption{QPS vs. recall, where (x) in the title of each subplot denotes the average query selectivity is $x$.}
    \label{fig:query}
\end{figure*}

\begin{figure}[t]
    \includegraphics[width=1.\linewidth]{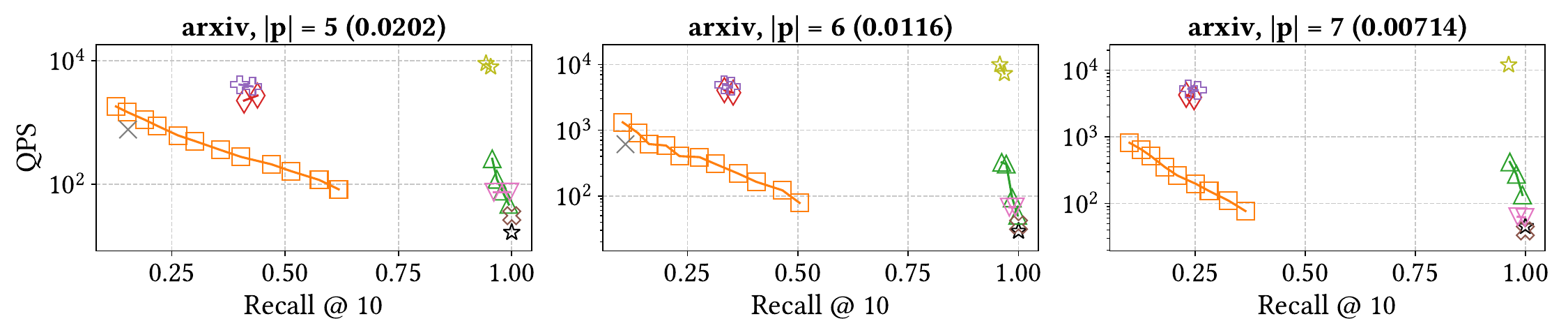}
    \vspace{-8mm}
    \caption{Long pattern queries on ArXiv dataset.}
    \label{fig:long-pattern}
\end{figure}

\begin{figure*}[t]
    \centering
    \includegraphics[width=\linewidth]{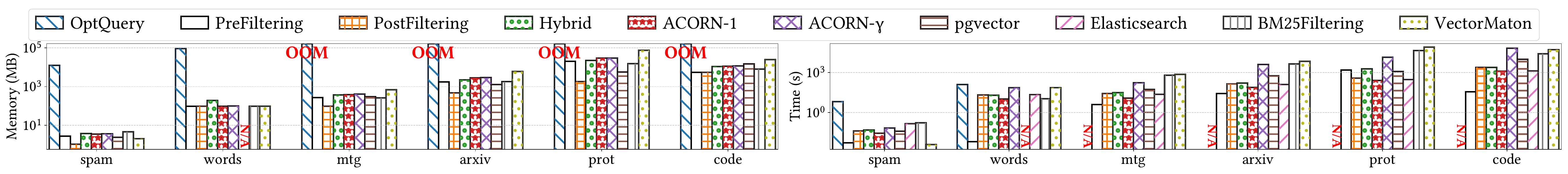}
    \vspace{-8mm}
    \caption{The peak memory consumption and time overhead of index construction.}
    \label{fig:memory-consumption}
\end{figure*}
\subsection{Query performance}
\label{sec:exp-query}
\noindent\textbf{Overall performance.} Figure~\ref{fig:query} reports the query performance of all methods. The $x$-axis \emph{recall} is defined in Section \ref{sec:problem-definition} (higher means better). Overall, \textsc{VectorMaton} consistently achieves one of the best recall-QPS trade-offs among all evaluated methods. When OptQuery is feasible, \textsc{VectorMaton} delivers comparable performance while requiring substantially lower index overhead. Compared with the remaining baselines, \textsc{VectorMaton} typically achieves both higher recall and higher query throughput. For example, on the CodeSearchNet dataset with pattern length $|p|=4$, \textsc{VectorMaton} achieves recalls close to 1, while simultaneously delivering up to two orders of magnitude higher QPS. Similar trends can be observed on all datasets, where \textsc{VectorMaton} is consistently located on the Pareto frontier of the recall-QPS trade-off. 

\vspace{0.5mm}
\noindent\textbf{Discussion on baseline performances.} We observe that the Hybrid method, which switches between PreFiltering and PostFiltering according to the query selectivity, consistently outperforms both individual methods. However, this improvement comes at the cost of higher indexing overhead. BM25Filtering introduces even higher indexing overhead while providing limited query-performance benefits. This suggests that a straightforward integration of inverted-index-based filtering with HNSW is insufficient for pattern-constrained vector search. Although the inverted index can effectively identify sequences sharing common $n$-grams with the query pattern, the HNSW graph associated with each $n$-gram posting list still contains many vectors whose sequences do not satisfy the full pattern constraint. Consequently, BM25Filtering suffers from a large number of false-positive candidates, leading to limitations similar to those of PostFiltering. ACORN generally achieves high query efficiency. However, because it performs filtering during graph traversal, a large fraction of vertices may be excluded under low-selectivity workloads,  reducing the connectivity of the searchable graph and leading to a low recall. Although ACORN mitigates this issue by increasing vertex degrees and exploring additional 2-hop neighbors, these techniques cannot fundamentally overcome the sparsity induced by aggressive filtering. For example, on the SwissProt dataset with $|p|=4$, the highest recall achieved by ACORN remains below 0.25. Pgvector adopts a hybrid strategy combining pre-filtering and post-filtering, and can occasionally achieve recall close to 1 under certain parameter settings. However, this comes at the expense of substantially lower query throughput. We also observe that adjusting \texttt{ef\_search} for BM25Filtering yields no change in recall or throughput because the final candidate selection is largely dominated by BM25 ranking rather than ANN search quality, such that some ground-truth results are filtered out during the BM25 ranking process. This indicates that BM25 may provide a negative impact in this problem setting as it is more suitable for textual relevance rather than exact pattern containment.

\vspace{0.5mm}
\noindent\textbf{Effects of varying $|p|$.} To evaluate the impact of pattern selectivity, we vary the pattern length $|p|\in\left\{2,3,4,5,6,7\right\}$ for each dataset. For short patterns $p\in[2,4]$, the results are shown in Figure~\ref{fig:query}. As $|p|$ increases, the number of vectors satisfying the pattern constraint generally decreases, leading to lower selectivity. Under this setting, most baseline methods experience performance degradation, as they must expand the search space to compensate for the higher fraction of filtered-out results and maintain recall. In contrast, \textsc{VectorMaton} and PreFiltering benefit from larger $|p|$. A longer pattern reduces the size of the filtered vector set $V_p$, thereby shrinking the effective search space of \textsc{VectorMaton} and PreFiltering. As a result, both \textsc{VectorMaton} and PreFiltering achieves improved query efficiency while maintaining high recall. For longer patterns with $|p| \in \{5,6,7\}$, we report the results on the ArXiv dataset in Figure~\ref{fig:long-pattern}; results for the remaining datasets are included in the appendix of our extended version~\cite{extended}. We observe that the selectivity becomes extremely low for these pattern lengths. For example, when $|p|=5$, the average selectivity is already below 1\% on most datasets. Consequently, post-filtering-based approaches (e.g., PostFiltering and ACORN) suffer from severe recall degradation. Although PreFiltering achieves slightly higher query efficiency than in the short-pattern setting, it remains substantially slower than \textsc{VectorMaton}, while \textsc{VectorMaton} continues to maintain both high recall and superior query throughput under these workloads. We also remark that for even longer patterns $|p|>7$, the selectivity is too low on all datasets and there are too few satisfying vectors, and therefore we omit their results.

\vspace{0.5mm}
\noindent\textbf{Pattern distribution analysis.} We further analyze the distribution of query pattern selectivity in our workloads. Detailed results and analysis are deferred to the appendix of our extended version \cite{extended}.

\begin{figure}[t]
    \includegraphics[width=1.\linewidth]{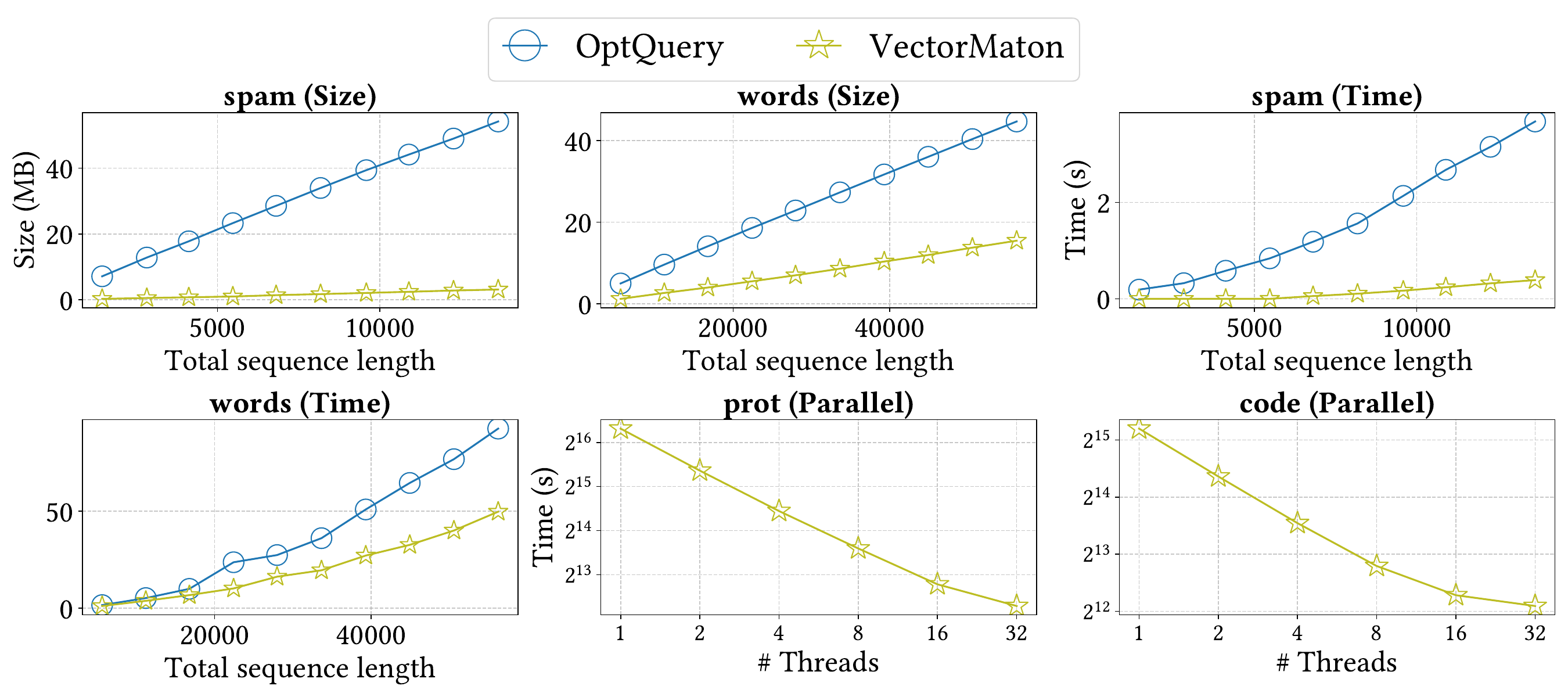}
    \vspace{-8mm}
    \caption{Index scalability test and parallelized index construction time.}
    \label{fig:scalability}
\end{figure}
\begin{figure}[t]
    \centering
    \includegraphics[width=1.\linewidth]{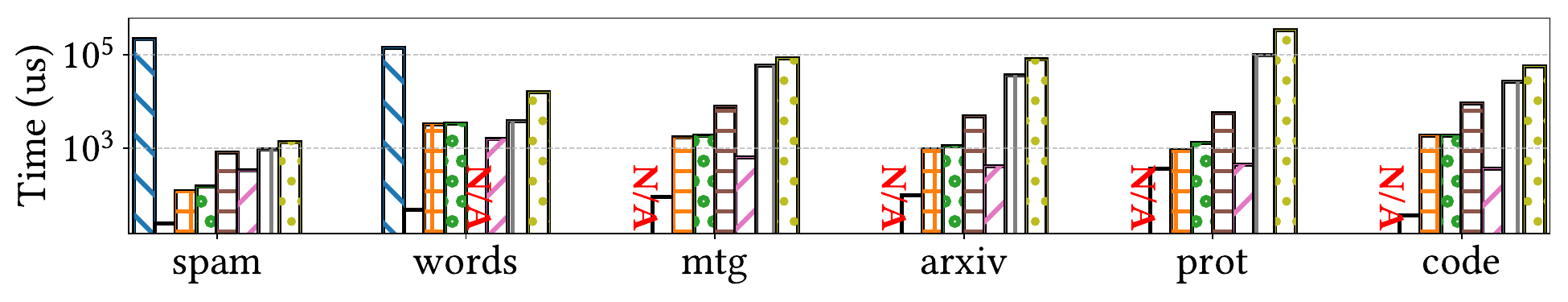}
    \vspace{-8mm}
    \caption{The average insertion time of 80\% base dataset and 20\% insertion data (ACORN-1 and ACORN-$\gamma$ are excluded).}
    \label{fig:insertion}
\end{figure}
\subsection{Index construction and maintenance}
\label{sec:exp-index}
\noindent\textbf{Memory overhead.} We evaluate the peak memory consumption during data preparation and index construction in Figure~\ref{fig:memory-consumption}. We note that Elasticsearch runs within a JVM configured with a fixed pre-allocated heap size, making its actual memory usage difficult to measure fairly; therefore, we exclude it from this comparison. PostFiltering and pgvector exhibit the lowest memory consumption because they only build a single HNSW index without any adjustments. However, this simplicity comes at the cost of low recall or efficiency. For example, on SwissProt dataset with $|p|=4$, pgvector can only answer less than 10 queries per second, making them impractical for actual use. To accelerate pattern filtering, we integrate PreFiltering, Hybrid, and ACORN with our ESAM implementation. As a result, these methods require additional memory beyond their original designs. \textsc{VectorMaton} incurs higher memory consumption than most baseline methods due to the need to build multiple HNSW graphs based on the sequences, but this overhead is compensated by substantially superior query performance. On small-to-moderate datasets, \textsc{VectorMaton} requires an affordable memory of less than 10GB to build the index.

\vspace{0.5mm}
\noindent\textbf{Time overhead.} Figure~\ref{fig:memory-consumption} reports the index construction time of all methods. PreFiltering, PostFiltering, Hybrid, ACORN-1, pgvector, and Elasticsearch exhibit relatively low construction overhead because they either require negligible indexing structures (PreFiltering) or maintain only a single HNSW graph (the remaining methods). ACORN-$\gamma$ incurs substantially higher construction time than ACORN-1 due to the additional graph expansion required to increase vertex connectivity and incorporate multi-hop neighborhood information. In contrast, OptQuery, BM25Filtering, and \textsc{VectorMaton} incur higher construction overhead because they maintain multiple HNSW graphs to support pattern-constrained search. Among them, OptQuery and \textsc{VectorMaton} effectively trade index construction time for significantly improved query performance. BM25Filtering also incurs considerable construction overhead because it builds HNSW indexes for the posting lists associated with individual $2$-grams. However, unlike OptQuery and \textsc{VectorMaton}, this additional indexing cost does not translate into corresponding query-performance gains. Overall, the results indicate that \textsc{VectorMaton} achieves a favorable trade-off between index construction overhead and query efficiency, whereas BM25Filtering introduces substantial indexing costs while providing limited benefits during query processing. We also remark that the construction overhead of \textsc{VectorMaton} can be substantially reduced through parallel index construction (see Section \ref{sec:exp-case-study} ``parallel index construction'').

\vspace{0.5mm}
\noindent\textbf{Index scalability test.}
We evaluate the scalability of index size and construction time on Spam and Words datasets for OptQuery and \textsc{VectorMaton} methods. For each dataset, we construct the index over progressively larger subsets $\left\{10\%, 20\%,\cdots,100\%\right\}$ of the data and plot the resulting size-length and time-length curves. Figure \ref{fig:scalability} (spam, words) illustrates the results. Although the theoretical worst-case space complexity is $O(m^2)$ for OptQuery and $O(m^{1.5})$ for \textsc{VectorMaton}, the empirical space costs exhibit near-linear growth in practice for both methods. This behavior arises because real sequences typically contain repeated substrings, resulting in fewer distinct patterns to index than the worst-case bound. In addition, we observe the growth rate of \textsc{VectorMaton} is significantly lower than that of OptQuery. As the dataset size increases, the gap between the two methods becomes larger. These results demonstrate that \textsc{VectorMaton} scales gracefully with increasing data size and can effectively handle large-scale sequence collections.

\vspace{0.5mm}\noindent\textbf{Index insertion.} We evaluate the index insertion overhead of all methods. Specifically, for each dataset, we first construct the index on 80\% of the data and subsequently insert the remaining 20\%. ACORN-1 and ACORN-$\gamma$ do not support dynamic updates and are therefore excluded from this experiment. Since the memory overhead during insertion is similar to that during index construction, we focus on insertion time in this evaluation. Figure~\ref{fig:insertion} reports the average insertion time of all methods. Overall, we observe that the insertion overheads of most methods are proportional to their corresponding index construction overheads, as most methods inherently construct their index via an incremental index maintenance. For \textsc{VectorMaton}, the proposed insertion scheme is the amortized cost of the index construction time, and does not introduce additional asymptotic complexity. Consequently, its insertion overhead remains moderate. Among the baselines, PreFiltering incurs the lowest insertion cost because it does not maintain any HNSW graph. Elasticsearch also exhibits relatively low insertion overhead due to its segment-based architecture. Newly inserted vectors are only added to the HNSW graph of a newly created segment rather than a global HNSW graph over the entire dataset. As a result, insertion cost depends primarily on the segment size rather than the total dataset size. However, query processing must search multiple segment-level HNSW graphs and merge their results, while periodic segment merges incur additional maintenance overhead.

\begin{table}[t]
\centering
\caption{Reduction of index size and construction time using proposed index strategies.}
\vspace{-4mm}
\resizebox{0.9\linewidth}{!}{
\begin{tabular}{l|c|c|c|c|c|c}
\hline
\textbf{Datasets} & spam & words & mtg & arxiv & prot & code \\ 
\hline
\textbf{Reduction (size)} & 93.9\% & 75.7\% & 80.6\% & 77.6\% & 72.9\% & 67.2\% \\ 
\hline\hline
\textbf{Datasets} & spam & words & mtg & arxiv & prot & code\\ 
\hline
\textbf{Reduction (time)} & 99.4\% & 6.7\% & 58.4\% & 64.7\% & 22.2\% & 27.0\% \\ 
\hline
\end{tabular}}
\label{tab:ablation}
\end{table}
\subsection{Case studies}
\label{sec:exp-case-study}
\noindent\textbf{Parallelized index construction.}
We evaluate the parallel index construction on two large-scale datasets, SwissProt and CodeSearchNet. For each dataset, we construct the index using 1, 2, 4, 8, 16, and 32 threads, respectively. Figure~\ref{fig:scalability} (prot, code) reports the results. The results show that \textsc{VectorMaton} scales effectively in a parallel setting, achieving linear speedup when the number of threads does not exceed 16. Beyond this point, the speedup decreases. This is because our parallelization strategy distributes the construction of different state indexes across threads, while the construction of each individual index remains sequential. As the number of threads increases, the overall runtime becomes dominated by a few large state indexes whose construction cannot be further parallelized, leading to reduced scalability gains. Nevertheless, with 16 threads, the index construction time is reduced to less than two hours on both datasets, which is efficient for large-scale datasets.

\vspace{0.5mm}
\noindent\textbf{Ablation study of index strategies.}
Recall that \textsc{VectorMaton} incorporates the index reuse and skip-build strategies. We conduct an ablation study to evaluate their effects with the results in Table~\ref{tab:ablation}, showing that they effectively reduce both index size and time, with an up to 93.9\% reduction in size and 99.4\% reduction in time. The reason is that index reuse strategy reduces redundancy of overlapping vector subsets across related states, while the skip-build strategy avoids unnecessary HNSW construction for small ID sets. In the skip-build strategy, we also introduce a threshold parameter $T$ and we evaluate the impact of varying it. Detailed results and analysis are deferred to the appendix of our extended version \cite{extended}.

\vspace{0.5mm}
\noindent\textbf{Parameter tuning of PostFiltering.}
Recall that in PostFiltering (Algorithm~\ref{alg:filtering}), HNSW search is performed with parameter $ef\_search$, and the top-$ef\_search$ candidates are retrieved before applying pattern filtering. We further study the impact of the candidate set size on query performance. The corresponding results and analysis are provided in the appendix of our extended version \cite{extended}.

\begin{figure}[t]
    \centering
    \includegraphics[width=\linewidth]{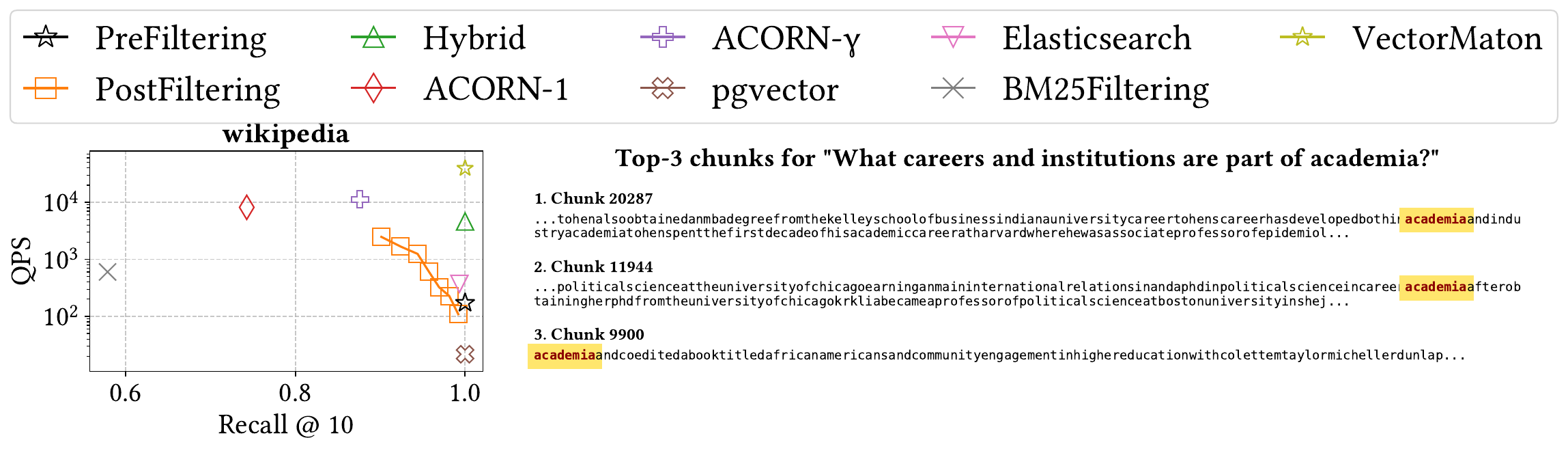}
    \vspace{-8mm}
    \caption{Case study of Wikipedia.}
    \label{fig:wikipedia}
\end{figure}
\subsection{Applicable long-sequence scenario}
\label{exp:long-sequence}
Practical scenarios may involve longer sequences than those in our datasets. One representative example is retrieval-augmented generation (RAG) systems~\cite{rag}, where documents are partitioned into chunks and represented as vectors for semantic retrieval. In practice, users often combine semantic search with keyword or phrase constraints (e.g., requiring a retrieved chunk to contain a specific entity or term). Such queries naturally correspond to the ANNS with pattern constraints problem studied in this paper.

We conduct an additional case study using Wikipedia articles~\cite{wikipediaarticles}. Specifically, we select 10,000 articles and partition them into text chunks, using their embeddings as vectors and the original text as associated sequences. The resulting dataset contains approximately
20 million characters, with many sequences exceeding 1,000 characters. We further construct 75 realistic query workloads consisting of an embedded prompt and a keyword constraint generated by ChatGPT~\cite{chatgpt2026}. For example, the prompt ``What careers and institutions are part of academia?'' is paired with the keyword constraint ``academia''. Figure~\ref{fig:wikipedia} reports the results. We observe that \textsc{VectorMaton} achieves the highest query throughput while achieving recall 1.0, demonstrating its effectiveness on long-sequence workloads representative of practical RAG applications. The top-ranked result for the example query discusses the educational background, academic positions, and affiliated institutions of a researcher, while explicitly containing the keyword ``academia''. This demonstrates that \textsc{VectorMaton} successfully retrieves semantically relevant content that also satisfies the specified keyword constraint. While practical applications may involve more flexible constraints, such as accepting terms similar in meaning to the specified keyword, we leave support for such constraints to future work.

\section{Related works}
\label{sec:related}
\textbf{ANNS approaches.} ANNS methods can be generally classified into four categories: (1) tree-based methods~\cite{kdtree,balltree,ringcovertree}; (2) hashing-based methods~\cite{pm-lsh,lshann2,det-lsh,simhash}; (3) quantization-based methods~\cite{pqscan,treepq,rabitq,rabitqenhanced}; and (4) graph-based methods~\cite{hnsw,nsg,diskann}. Tree-based methods partition the data space hierarchically and organize points into recursive structures, including KD-tree~\cite{kdtree}, ball tree~\cite{balltree}, and ring-cover tree~\cite{ringcovertree}. Hashing-based methods include locality-sensitive hashing (LSH)~\cite{pm-lsh,lshann2,det-lsh} and SimHash \cite{simhash}. Quantization-based methods \cite{pqscan,treepq,rabitq,rabitqenhanced} encode high-dimensional vectors into compact representations using learned codebooks and approximate distances in the compressed space. Graph-based methods construct a proximity graph (e.g., HNSW \cite{hnsw}, NSG \cite{nsg}) over the dataset and perform query processing via greedy graph traversal.

\vspace{0.5mm}
\noindent\textbf{Filtered ANNS approaches.} 
Existing filtered ANNS approaches mainly focus on attribute-based filtering~\cite{unify,attributeannsurvey,filtereddiskann,ivf-subset} and range-based filtering~\cite{rangepq,dynamicrfann,serf,digra,rangepq,irangegraph}. Attribute-based filtering includes selecting vectors whose categorical attribute equals a queried value~\cite{filtereddiskann}, or whose attribute set contains the queried set~\cite{unify}. Range-based filtering considers numerical constraints associated with each vector \cite{rangepq,dynamicrfann,serf,digra,irangegraph}. The objective is to answer ANN queries while restricting results to vectors whose associated numeric value falls within a query interval. Recently, a few studies (e.g., ACORN \cite{acorn}) have investigated general filtered ANNS, with some \cite{vecbench,pathseer} appeared contemporaneously with this paper.
\section{Conclusion}
In this paper, for the first time, we introduce the problem of ANNS with pattern constraints. To address this problem, we propose \textsc{VectorMaton}, an index that jointly supports similarity search and pattern filter. Experimental results show that \textsc{VectorMaton} achieves a favorable trade-off among query latency, recall, and index space consumption compared with baseline methods.
\begin{acks}
This research / project is supported by the National Research Foundation, Singapore, under its Frontier CRP Grant (NRF-F-CRP-2024-0005), and under its AI Singapore Programme (AISG Award No: AISG3-RP-2024-034). This research is also supported by Singapore MOE AcRF Tier-2 Award (MOE-T2EP20224-0005) and NTU SUG-NAP. Any opinions, findings, and conclusions or recommendations expressed in this material are those of the author(s) and do not reflect the views of the National Research Foundation, Singapore.
\end{acks}

\bibliographystyle{ACM-Reference-Format}
\balance
\bibliography{reference}
\clearpage
\appendix

\section{Pseudo codes of OptQuery}
\begin{algorithm}[ht]
\small
\caption{OptQuery}
\label{alg:pc-hnsw}

\SetKwProg{Fn}{Function}{:}{}
\SetKwFunction{Build}{Build}
\SetKwFunction{Query}{Query}

\Fn{\Build{$n,V,S$}}{
    $H\gets$ a hashmap that maps patterns to HNSW graphs\;
    \For{$i\gets1,2,\cdots n$}{
        \For{$j\gets 1,2,\cdots,|s_i|$}{
            \For{$k\gets j,j+1,\cdots,|s_i|$}{
                Extract substring $s_i[j..k]$ as a pattern $p$\;
                Insert $v_i$ to the graph of $H(p)$\;
            }
        }
    }
}

\BlankLine

\Fn{\Query{$v_q, p, k, ef\_search$}}{
    \lIf{$H(p)$ does not exist}{
        \textbf{return} $\emptyset$
    }
    Search in the HNSW graph $H(p)$ with parameters $v_q,k,ef\_search$ and return $k$ results\;
}

\end{algorithm}

\section{Proof of Lemma \ref{lemma:num-states}}
Let $U$ be the set of all $(id,pos)$ pairs in the
sequence collection. Since the total sequence length is $m$, we have
$|U| = m$. Let $\mathcal{F}$ denote the family of all distinct nonempty position
lists of patterns appearing in the sequence collection. We first show
that $\mathcal{F}$ is a laminar family. That is, for any
$P,Q \in \mathcal{F}$, either $P \cap Q = \emptyset$, $P\subseteq Q$ or $Q \subseteq P$.

Consider two patterns $p$ and $q$ and let $P = \operatorname{poslist}(p), Q = \operatorname{poslist}(q)$, and suppose that $P \cap Q \neq \emptyset$. Then there exists a pair
$(id,pos) \in P \cap Q$, meaning that both $p$ and $q$ occur in
sequence $s_{id}$ and end at the same position $pos$. Therefore, both
$p$ and $q$ are suffixes of the same prefix $s_{id}[0..pos]$.

Without loss of generality, assume that $|p| \le |q|$. Since both
patterns end at the same position, $p$ must be a suffix of $q$.
Consequently, every occurrence of $q$ ending at some pair
$(id',pos')$ also induces an occurrence of $p$ ending at the same
pair. Hence,
\[
\operatorname{poslist}(q)
\subseteq
\operatorname{poslist}(p),
\]
which proves that $\mathcal{F}$ is laminar.

We next augment $\mathcal{F}$ by adding $U$ and all singleton sets
$\{x\}$ for $x \in U$. Let the resulting family be $\mathcal{F}'$, i.e.,
$$\mathcal{F}'=\mathcal{F}\cup\left\{U\right\}\cup\left\{\left\{x\right\}:x\in U\right\}$$
Clearly, $\mathcal{F}'$ remains laminar. We organize the sets in
$\mathcal{F}'$ into an inclusion tree rooted at $U$, where the parent
of each set is its smallest strict superset in $\mathcal{F}'$.

The $m$ singleton sets are exactly the leaves of this tree. Moreover,
every internal node has at least two children. Otherwise, suppose an
internal node $A$ had only one child $B \subsetneq A$. Since
$B \subsetneq A$, there exists some $x \in A \setminus B$. The
singleton $\{x\}$ belongs to $\mathcal{F}'$ and satisfies
$\{x\} \subset A$, so it must be contained in a child of $A$ distinct
from $B$, yielding a contradiction.

A rooted tree with $m$ leaves in which every internal node has at
least two children contains at most $2m-1$ nodes. Therefore,
\[
|\mathcal{F}|
\le
|\mathcal{F}'|
\le
2m-1
=
O(m).
\]

Therefore, the number of distinct position list, equivalence class, and maximal patterns are bounded by $O(m)$. Thus, the total number of
states in ESAM is also $O(m)$.

\section{Proof of Lemma \ref{lemma:num-id-set}}
Consider a sequence $s$ of length $|s|$. It has at most $\frac{|s|(|s|+1)}{2}$ distinct substrings, and thus can contribute its sequence ID to at most that many maximal patterns. By Lemma~\ref{lemma:num-states}, the total number of states in \textsc{VectorMaton} is bounded by $O(m)$; denote this bound by $cm$ for some constant $c$. Therefore, the ID $i$ of sequence $s_i$ can appear in at most: $\min\left\{\frac{|s_i|(|s_i|+1)}{2},cm\right\}$ states. Let $S$ denote the sequence collection. The total size of all associated ID sets is thus bounded by $\sum_{i=1}^n\min\left\{\frac{|s_i|(|s_i|+1)}{2},cm\right\}$. Next, we denote $t_i=\min\left\{\frac{(|s_i|+1)^2}{2},cm\right\}\geq\min\left\{\frac{|s_i|(|s_i|+1)}{2},cm\right\}$ and consider two cases.
\begin{itemize}[leftmargin=*]
    \item If $\frac{(|s_i|+1)^2}{2}\leq cm$, then $|s_i|+1\leq\sqrt{2cm}$ and $t_i=\frac{(|s_i|+1)^2}{2}\leq (|s_i|+1)\cdot\frac{\sqrt{2cm}}{2}$.
    \item If $\frac{(|s_i|+1)^2}{2}>cm$, then $\sqrt{m}<\frac{(|s_i|+1)\sqrt{2c}}{2c}$ and $t_i=cm<(|s_i|+1)\cdot\frac{\sqrt{2cm}}{2}$.
\end{itemize}

Therefore, the total size of all associated ID sets is bounded by:

$$\sum_{i=1}^nt_i\leq \sum_{i=1}^n(|s_i|+1)\frac{\sqrt{2cm}}{2}=\frac{\sqrt{2c}}{2}(m^{1.5}+n\sqrt m)=O(m^{1.5})$$

\section{Proof of Lemma \ref{lemma:transition}}
Let $poslist(p_A)=\left\{(id_j,pos_j)\right\}_{j=1}^k$. Since $p_B = p_A \cdot c$, an occurrence of $p_A$ at position $(id_j, pos_j)$ can be extended by $c$ if and only if $pos_j < |s_{id_j}| - 1$ and the next character is $c$. Therefore, $poslist(p_B)=\left\{(id_j,pos_j+1)\mid 1\leq j\leq k,s_{id_j}[pos_j+1]=c\right\}$. For any other pattern $p_A' \in A$, by definition of equivalence classes we have $\mathit{poslist}(p_A') = \mathit{poslist}(p_A)$. Hence, extending $p_A'$ by the same character $c$ yields the same set of occurrences and thus the same position list as $p_B$. Consequently, $p_A' \cdot c$ belongs to the same equivalence class $B$, which establishes a well-defined transition from state $A$ to state $B$ labeled by $c$.

\section{Proof of Lemma \ref{lemma:exact-cover}}
In index reuse strategy, we construct $I_j$ only over the difference set $V_j\setminus I_k$. Therefore, it directly draws the result that $I_j\cup I_k=V_j$ and $I_j\cap I_k=\emptyset$.

\section{Proof of Lemma \ref{lemma:case1}}
Let $A$ be a suffix state with position list $poslist(A)={(x_k,y_k)}$.
If there is no outgoing transition from $A$ labeled $c$, then for every $(x_k,y_k)\in poslist(A)$, either $y_k+1$ exceeds the length of $s_{x_k}$ or $s_{x_k}[y_k+1]\neq c$. Hence, none of the existing occurrences of patterns in $A$ can be extended by $c$. Since we are currently processing position $(i,j)$ and the next symbol is $c=s_i[j+1]$, appending $c$ introduces a new occurrence ending at $(i,j+1)$. Therefore, extending $A$ by $c$ yields a new equivalence class with position list $\left\{(i,j+1)\right\}$.

\section{Suffix link and analysis}
\subsection{Correctness}
In this section, we will prove the correctness of our automaton construction algorithm in an on-line faction (i.e., the correctness of the incremental process). We first present the definition of the \emph{suffix link}~\cite{automaton}. While it is originally introduced for the single-sequence setting, the concept naturally extends to our multi-sequence scenario.
\begin{definition}[Suffix link]
Given a state $A$ whose maximal pattern is $s$, the suffix link $\mathrm{link}(A)$ points to the state whose maximal pattern is the longest proper suffix of $s$ that is not equivalent to $s$.
\end{definition}

We next show that all suffix states can be organized into a chain of strictly decreasing maximal pattern lengths via suffix links.

\begin{lemma}
\label{lemma:suffix-chain}
Consider sequence $s_i$ currently being processed, and suppose the algorithm has processed position $j$. Let $A$ be the state corresponding to pattern $s_i[0..j]$. Then the set of all suffix states is precisely $\left\{A,\mathrm{link}(A),\mathrm{link}(\mathrm{link}(A)),\cdots,0\right\}$, where $0$ denotes the initial state.
\end{lemma}
\begin{proof}
By definition of suffix states, a suffix state is a state whose position list is empty or contains $(i,j)$. Hence, every suffix state represents at least a pattern that is a suffix of $s_i[0..j]$. Moreover, the maximal pattern in a suffix state must be also a suffix of $s_i[0..j]$, which is not hard to see.

Let $A$ be the state corresponding to $s_i[0..j]$. Then $A$ is a suffix state with maximal pattern $s_i[0..j]$. By construction of suffix links, $\mathrm{link}(A)$ points to the state representing the longest proper suffix of the maximal pattern of $A$ that belongs to a different equivalence class. Repeatedly following suffix links therefore enumerates states whose maximal patterns are exactly the suffixes of $s_i[0..j]$ in strictly decreasing order of length, until reaching the initial state $0$, which represents the empty pattern. Each of these states contains $(i,j)$ in its position list (or is the initial state), and is thus a suffix state.

Conversely, let $C$ be any suffix state. Then $C$ represents a pattern that is a suffix of $s_i[0..j]$. By the defining property of suffix links in the (generalized) suffix automaton, every suffix of $s_i[0..j]$ is represented by a state on the suffix-link chain starting from $A$. Therefore, $C$ must belong to the set $\left\{A,\mathrm{link}(A),\mathrm{link}(\mathrm{link}(A)),\cdots,0\right\}$.
\end{proof}

Next, we show that the suffix-link chain of suffix states can be partitioned into at most two contiguous segments: (1) states without an outgoing transition labeled $c=s_i[j+1]$, and (2) states with such a transition.

\begin{lemma}
\label{lemma:two-new-states}
Let $A$ be the state corresponding to pattern $s_i[0..j]$, and let $c=s_i[j+1]$ be the next symbol to be processed. Then the suffix-link chain $\left\{A,\mathrm{link}(A),\mathrm{link}^2(A),\cdots,0\right\}$ can be divided into at most two contiguous segments such that:

\begin{itemize}[leftmargin=*]
\item In the first segment (if it exists), no state has an outgoing transition labeled $c$.
\item In the second segment (if it exists), every state has an outgoing transition labeled $c$.
\end{itemize}
\end{lemma}

\begin{proof}
Let $A$ be the state corresponding to the pattern $s_i[0..j]$, and consider the suffix-link chain $\left\{A, \mathrm{link}(A), \mathrm{link}^2(A), \ldots, 0\right\}$.

If $A$ has an outgoing transition labeled $c$, then the extended pattern $s_i[0..j+1]$ is already represented in the successor state. By the construction of suffix links, all states reachable via $\mathrm{link}$ from $A$ correspond to shorter suffixes of $s_i[0..j]$. Since $s_i[0..j+1]$ contains these suffixes, each of these states must also have an outgoing transition labeled $c$.

Otherwise, if $A$ has no transition labeled $c$, we follow the suffix links until we reach the first state $B$ that does have a $c$-transition (if any). All states before $B$ lack such a transition, while all states from $B$ onward have it. If no state along the chain has a $c$-transition, then the entire chain forms a single segment without a $c$-transition.
\end{proof}

For the first segment, by Lemma~\ref{lemma:case1}, we create a new state and connect all states in that segment to it via transitions labeled $c$. For the second segment, we only need to examine the first state in the segment.

\begin{lemma}
Let $B$ be the first state in the second segment of the suffix-link chain, and let $C$ be the successor of $B$ via the transition labeled $c$. Let $s(B)$ and $s(C)$ denote the maximal patterns of states $B$ and $C$, respectively. If $|s(C)| = |s(B)| + 1$, then extending by $c$ does not introduce any structural change in the second segment.
\end{lemma}

\begin{proof}
We first establish the following claim.

\textbf{Claim.}
Let $X$ be a state with maximal pattern $s(X)$. Then every pattern $y$ in the equivalence class of $X$ is a suffix of $s(X)$.

\textbf{Proof of claim.}
Let $(i,j)$ be any occurrence in the position list of $s(X)$. Since all patterns in state $X$ share the same position list, the pattern $y\in X$ must also end at position $j$ in sequence $s_i$. Hence, $s(X)=s_i[j-|s(X)|+1..j]$ and $y=s_i[j-|y|+1..j]$, which shows that $y$ is a suffix of $s(X)$.

Now consider the main statement. Since $B$ has a $c$-transition to
$C$, by Lemma~\ref{lemma:transition}, $s(B)c$ belongs to $C$.
By the preceding claim, $s(B)c$ is a suffix of $s(C)$. Since
$|s(C)|=|s(B)|+1=|s(B)c|$, we have $s(C)=s(B)c$. As $s(B)$ is
a suffix of $s_i[0..j]$, $s(C)$ is therefore a suffix of
$s_i[0..j+1]$. By the claim, all patterns represented in state $C$ are suffixes of
$s(C)$, and are therefore also suffixes of $s_i[0..j+1]$. Hence, when
$c$ is processed, every pattern in $C$ gains the same new occurrence
$(i,j+1)$. Since these patterns had identical position lists before
the extension, they remain equivalent afterward. Consequently, no
refinement of state $C$ is required.
\end{proof}

We now consider the remaining case where $|s(C)|>|s(B)|+1$. This implies that state $C$ represents patterns strictly longer than the direct extension of $s(B)$ by $c$. In this situation, the extension by $c$ introduces a refinement of the equivalence relation represented by $C$, and the state must be split.

\begin{lemma}
Let $B$ be the first state in the second segment of the suffix-link chain, and let $C$ be the successor of $B$ via the transition labeled $c$. Let $s(B)$ and $s(C)$ denote the maximal patterns of states $B$ and $C$, respectively. If $|s(C)| > |s(B)| + 1$, then for every state $B'$ in the second segment whose $c$-transition points to $C$, the equivalence class obtained by extending $B'$ with $c$ is no longer correctly represented by $C$.
\end{lemma}
\begin{proof}
Assume for contradiction that after processing $c$,
the extension $s(B)c$ is still represented by $C$. Since $B$ is a
suffix state, $s(B)c$ occurs at the new position $(i,j+1)$. Thus,
$s(C)$ must also occur at $(i,j+1)$; otherwise, $s(B)c$ and $s(C)$
would no longer have the same position list. Hence, let $P=s_i[0..j]$, and then $s(C)$ is a suffix of $Pc$. Write $s(C)=tc$. Then $t$ is a
suffix of $P$ and
\[
|t|=|s(C)|-1>|s(B)|
\]
Let $D$ be the state representing $t$. By Lemma~6, $D$ is a suffix
state on the suffix-link chain, and since $tc=s(C)$, $D$ has an
outgoing transition labeled $c$ to $C$ (by Lemma \ref{lemma:transition}). Moreover,
\[
|s(D)|\ge |t|>|s(B)|
\]
so $D$ appears before $B$ on the suffix-link chain. This contradicts
the definition of $B$ as the first suffix state having an outgoing
$c$-transition. Therefore, $s(C)$ does not occur at $(i,j+1)$. Now consider any state $B'$ in the second segment whose $c$-transition
points to $C$. The pattern $s(B')c$ gains the new occurrence
$(i,j+1)$, whereas $s(C)$ does not. Hence their position lists become
different, so $s(B')c$ can no longer be represented by $C$. The lemma
follows.
\end{proof}

After splitting state $C$ into $C$ and its clone $C'$, we maintain the invariant that
$\mathrm{poslist}(C') = \mathrm{poslist}(C) \cup \{(i,j+1)\}$, while $C$ retains its original position list. We now determine precisely which transitions must be redirected to $C'$.

Let $X$ be a suffix state in the second segment whose $c$-transition previously pointed to $C$. Then its maximal pattern occurs at position $(i,j)$. Extending this pattern by $c = s_i[j+1]$ therefore produces an occurrence ending at $(i,j+1)$. Consequently, the extended pattern has position list equal to $\mathrm{poslist}(C')$, and must be represented by state $C'$. Hence the outgoing $c$-transition of $X$ must be redirected to $C'$.

Conversely, let $Y$ be any state that is not a suffix state whose $c$-transition points to $C$. Extending its maximal pattern by $c$ does not yield an occurrence ending at $(i,j+1)$, and therefore its extended pattern does not include $(i,j+1)$ in its position list. The corresponding equivalence class thus remains $\mathrm{poslist}(C)$, and the transition of $Y$ continues to point to $C$. No redirection is required.

Therefore, exactly the suffix states in the second segment require their $c$-transitions to be redirected from $C$ to $C'$. Combining all above analysis, we conclude that after each extension step, the automaton maintains correct to represent all equivalence classes.
\subsection{Maintenance of suffix links}
As discussed above, processing each symbol creates at most two new states. We now describe how suffix links are maintained based on Algorithm 2. For a state created by cloning, its suffix link is initialized to the suffix link of the original state being cloned. If the first segment is empty, \(cur\) is an existing state and its suffix link remains unchanged. We therefore only prove below the suffix link of a new state created for a nonempty first segment. Specifically, if a new state \(cur\) is created for the first segment, its suffix link depends on whether the second segment exists:

\begin{itemize}[leftmargin=*]
\item If the second segment does not exist (i.e., no suffix state has an outgoing transition labeled $c$), then $\mathrm{link}(cur)$ is set to the initial state $0$.
\item Otherwise, let $B$ be the first state in the second segment, and let $C$ be its successor via the transition labeled $c$ after any necessary cloning has been performed. Then we set $\mathrm{link}(cur) = C$.
\end{itemize}

This maintains the invariant that the suffix link of a state points to the state representing the longest proper suffix of its maximal pattern.

\begin{lemma}
The suffix link $\mathrm{link}(cur)$ constructed as above is correct.
\end{lemma}

\begin{proof}
Let $s(cur)=s_i[0..j+1]$ denote the maximal pattern of the newly created state $cur$.

\textbf{Case 1:} The second segment does not exist.
This means that no suffix state of $s_i[0..j]$ has an outgoing transition labeled $c$. Hence, for every proper suffix $x$ of $s_i[0..j]$, the extension $x+c$ does not occur in the previously processed prefix. Consequently, among all proper suffixes of $s(cur)$, none except the empty pattern corresponds to an existing state in a different equivalence class. Therefore, the longest proper suffix of $s(cur)$ that belongs to a different equivalence class is the empty pattern, represented by the initial state $0$. Thus $\mathrm{link}(cur)=0$ is correct.

\textbf{Case 2:} The second segment exists.
Let $B$ be the first state in the second segment, and let $C$ be the successor of $B$ via transition labeled $c$ after any necessary structural updates. Since $B$ lies on the suffix-link chain of $s_i[0..j]$, $s(B)$ is a suffix of $s_i[0..j]$, and hence $s(B)+c$ is a proper suffix of $s(cur)$. Moreover, all suffix states preceding $B$ (which correspond to strictly longer suffixes of $s_i[0..j]$) do not have a transition labeled $c$. Therefore, no longer proper suffix of $s(cur)$ corresponds to an existing state in a different equivalence class. Thus $s(B)+c$ is the longest proper suffix of $s(cur)$ that belongs to a distinct equivalence class, and it is represented by state $C$. Hence, setting $\mathrm{link}(cur)=C$ satisfies the suffix-link invariant.
\end{proof}

\section{Proof of Theorem \ref{theorem:space-complexity}}
By Lemma~\ref{lemma:num-states}, the number of states in \textsc{VectorMaton} is bounded by $O(m)$, and by Lemma~\ref{lemma:num-id-set}, the total size of all ID sets is bounded by $O(m^{1.5})$. If the symbol set size is constant, each state has at most a constant number of outgoing transitions. Therefore, the total number of transitions is also bounded by $O(m)$, and the overall space consumption is dominated by the storage of ID sets and associated index structures, which is bounded by $O(m^{1.5})$.

\begin{figure*}[t]
    \includegraphics[width=1.\linewidth]{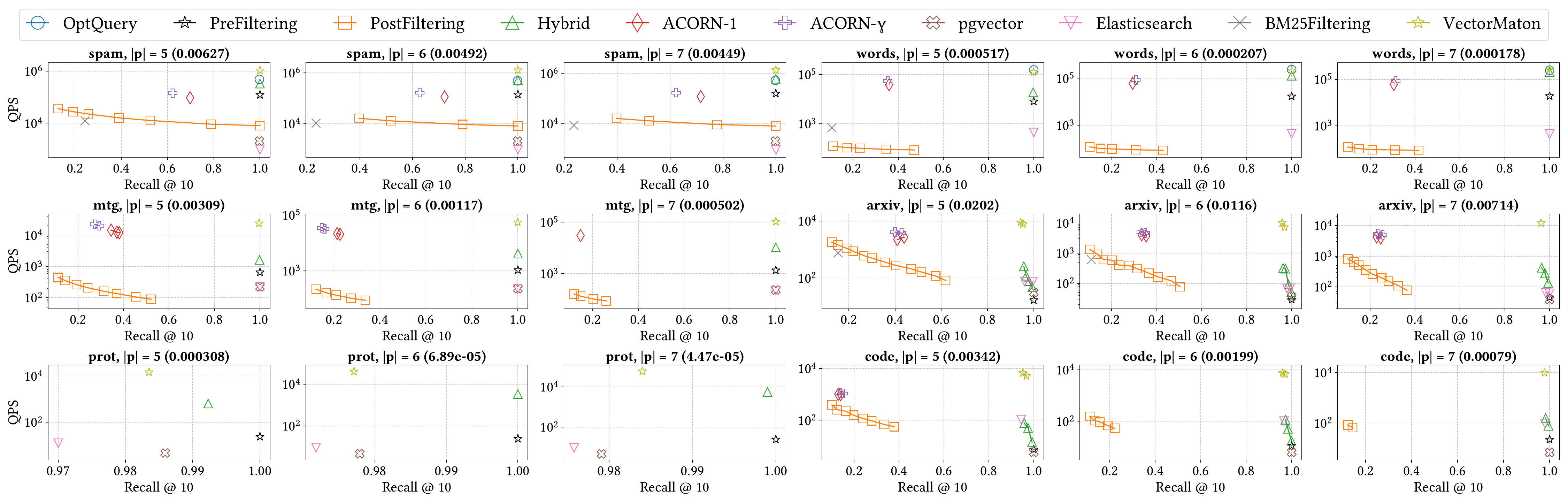}
    \caption{Query performance on longer patterns $|p|\in\left\{5,6,7\right\}$, where (x) in the title of each subplot denotes the average query selectivity is $x$.}
    \label{fig:long-pattern-full-appendix}
\end{figure*}

\section{Correctness of Lemma \ref{lemma:exact-cover} after insertion}
We show that the exact-cover property of Lemma~\ref{lemma:exact-cover} is preserved after the insertion process in Section \ref{sec:maintenance}. The correctness of the first two cases is straightforward. For the third case, let $I_j$ denote the local index of a state and $I_k$ denote its inherited index. By the above rules, when we process $I_j$, the newly inserted vector is added to $I_j$ if and only if it is not already covered by $I_k$. Therefore, the union $I_j \cup I_k$ must contain this new vector and covers all data associated with the state. It remains to show that $I_j \cap I_k = \emptyset$. Suppose, for contradiction, that both $I_j$ and $I_k$ contain the newly inserted vector. Then the insertion into $I_j$ must have occurred without observing the vector in $I_k$. However, updates are propagated recursively through inherited indexes, implying that $I_k$ is always updated before $I_j$ performs its coverage check. Consequently, $I_j$ would detect that the vector is already covered by $I_k$ and would not insert it, yielding a contradiction. Therefore, $I_j \cap I_k = \emptyset$, and Lemma \ref{lemma:exact-cover} is preserved after the insertion.

\section{Deletion of \textsc{VectorMaton}}
We discuss a lazy-deletion strategy that supports moderate deletion workloads. To enable deletions, we maintain, for each sequence, the set of ESAM states in which it appears. This information can be collected naturally during index construction and insertion, as each ESAM state already maintains the IDs of its associated sequences. To delete a sequence-vector pair with ID $i$, we first identify all states associated with the sequence using this auxiliary mapping. For each such state, we mark ID $i$ as deleted in the corresponding HNSW index without physically removing it from the graph structure. During query processing, any deleted IDs returned by HNSW search are filtered out before producing the final results.

This approach avoids expensive structural modifications to the HNSW indexes, and similar lazy-deletion mechanisms have been adopted in prior work on HNSW-based indexes (e.g.,~\cite{digra}). The trade-off is that deleted entries continue to occupy index space and may slightly degrade query performance as their number grows. Therefore, this strategy is most suitable for workloads with infrequent or moderate deletion rates.

\begin{figure}[t]
    \centering
    \includegraphics[width=1.\linewidth]{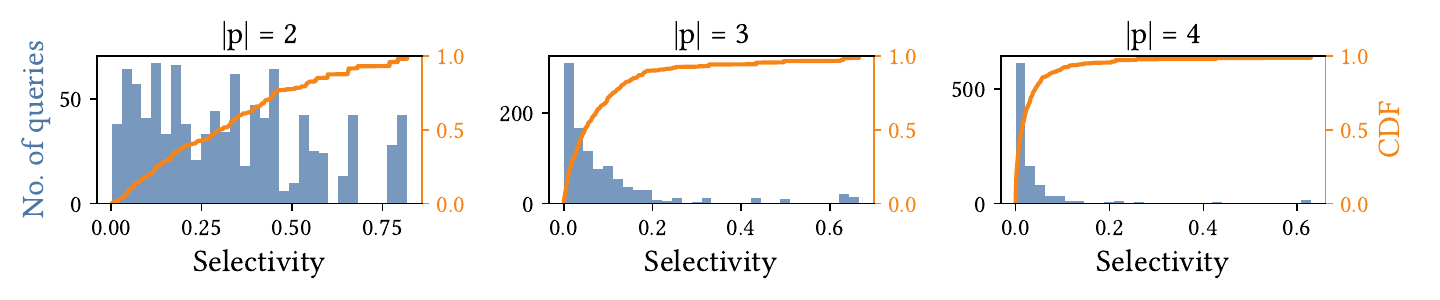}
    \caption{Query selectivity distribution on ArXiv dataset.}
    \label{fig:pattern-distribution}
\end{figure}
\begin{figure}[t]
    \centering
    \includegraphics[width=1.\linewidth]{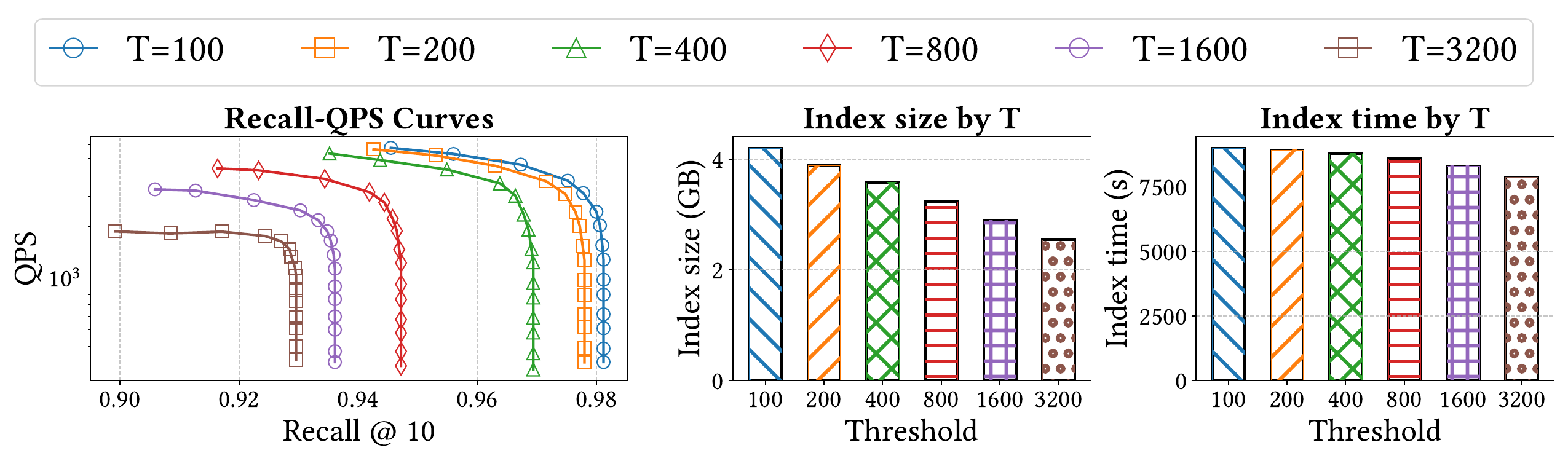}
    \caption{Parameter study of build threshold $T$.}
    \label{fig:threshold}
\end{figure}
\begin{figure}[t]
    \centering
    \includegraphics[width=\linewidth]{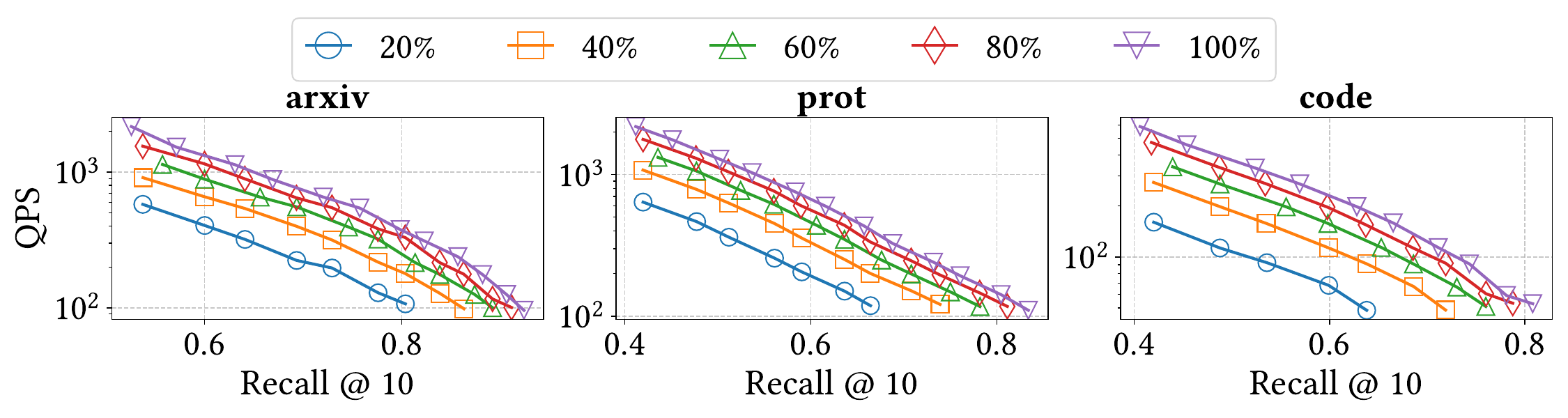}
    \caption{QPS-recall curves of PostFiltering with different candidate set size ratio.}
    \label{fig:postfiltering}
\end{figure}
\section{Dataset descriptions}
The detailed descriptions are provided below: (1) \textbf{Spam} \cite{spam}: a dataset of spam emails from the SpamAssassin corpus. Each record consists of an email title used as the sequence, and a pretrained embedding of the email content (generated by \texttt{all-MiniLM-L6-v2}) as the vector; (2) \textbf{Words} \cite{efarrall2024wordembeddings}: a dataset of word embeddings. Each record consists of a word (a sequence of letters) as the sequence and its pretrained word embedding as the vector; (3) \textbf{MTG} \cite{mtg}: a dataset of image embeddings. Each record consists of an image description as the sequence and a pretrained image embedding as its vector; (4) \textbf{ArXiv} \cite{arxiv}: a dataset of paper titles and embeddings. Each record consists of a paper title as the sequence and a text embedding (generated by \texttt{all-mpnet-base-v2}) as its vector; (5) \textbf{SwissProt} \cite{prot}: a subset of the UniProt protein sequence database. Each record consists of a protein sequence as the sequence and its structural embedding (generated by ProtBERT) as the vector; (6) \textbf{CodeSearchNet} \cite{code}: a dataset of code snippets. Each record consists of a function name as the sequence and a code embedding (generated by CodeBERT) as its vector.

\section{Query performance with longer query patterns}
Figure~\ref{fig:long-pattern-full-appendix} presents the complete results for longer query patterns with $|p| \in \{5,6,7\}$ across all datasets, extending the ArXiv results shown in Figure~\ref{fig:long-pattern} in the main paper. In some subplots, PostFiltering and ACORN results are omitted as their recalls are consistently less than 0.1. As the pattern length increases, the query selectivity decreases substantially, making the workloads increasingly challenging for post-filtering-based methods. Consequently, PostFiltering and ACORN generally exhibit significant recall degradation. Although PreFiltering benefits from the reduced candidate set size and achieves higher query efficiency than in the short-pattern setting, it remains substantially slower than \textsc{VectorMaton}. Overall, \textsc{VectorMaton} consistently maintains high recall while achieving superior query throughput across all datasets and pattern lengths.

\section{Pattern distribution analysis}
Figure~\ref{fig:pattern-distribution} presents the PDF and CDF of query selectivity on the ArXiv dataset, and we observe similar trends on the other datasets. Overall, the distribution tends to skew toward low-selectivity queries. As the pattern length increases, the distribution becomes increasingly concentrated in the low-selectivity region. In particular, the distribution for $|p|=2$ is substantially less skewed than those for $|p|=3$ and $|p|=4$. This observation explains why post-filtering-based methods remain competitive when $|p|=2$, but their recall deteriorates rapidly as $|p|$ increases and the query selectivity decreases.

\section{Effects of parameter $T$}
Figure~\ref{fig:threshold} illustrates the impact of varying $T$, where the queries are generated with mixed pattern lengths $|p|\in\left\{2,3,4\right\}$. As $T$ increases, both index size and construction time decrease, since more states skip graph construction. While increasing $T$ degrades query performance, query performance remains stable within a moderate range of $T$, indicating that brute-force search over small vector sets is sufficiently efficient.  This suggests a trade-off between index construction cost and query efficiency when selecting $T$.
\section{Parameter tuning of PostFiltering}
Our default choice of retrieving the top-$ef\_search$ candidates is motivated by the intuition that a larger candidate set provides more opportunities to retain at least $k$ valid results after pattern filtering, thereby improving recall. To evaluate the impact of the candidate set size, we keep the HNSW search parameter $ef\_search$ unchanged and retain only the top-$k_s$ candidates from the initial search, where
$k_s \in \{20\%,40\%,60\%,80\%,100\%\}\cdot ef\_search$.
Queries are generated with mixed pattern lengths $|p|\in\{2,3,4\}$. Figure~\ref{fig:postfiltering} reports the results. We observe that the recall--throughput trade-off consistently deteriorates as $k_s$ decreases. Therefore, retaining all $ef\_search$ candidates, i.e., $k_s=ef\_search$, provides the best overall query performance among the evaluated configurations.
\end{document}